\documentclass[
amsmath,amssymb,
twocolumn,
longbibliography,
superscriptaddress,
floatfix
]{revtex4-2}
\usepackage{tikz-cd}
\usepackage{bbm}
\usepackage{xcolor}
\usepackage{mathrsfs}
\usepackage[commandnameprefix=always]{changes}
\usepackage{epsfig}

\usepackage{soul,color}
\usepackage{graphicx}
\usepackage{mleftright}
\usepackage{amsfonts}
\usepackage{amsthm}
\usepackage[figuresright]{rotating}
\usepackage{amssymb}
\usepackage{amsmath}
\usepackage{dcolumn}
\usepackage{physics}
\usepackage{float}
\usepackage{bm}
\usepackage{verbatim}
\usepackage{braket}
\usepackage[normalem]{ulem}
\usepackage[ruled,vlined,linesnumbered]{algorithm2e}
\usepackage{setspace}
\usepackage[colorlinks,linkcolor=blue,anchorcolor=blue,citecolor=blue,urlcolor=blue]{hyperref}
\usepackage{stackengine}
\usepackage{subcaption}

\usepackage{tikz}
\usetikzlibrary{quantikz2}
\usepackage{adjustbox}

\usepackage{youngtab}

\newtheorem{theorem}{Theorem}

\newtheorem{corollary}{Corollary}[theorem]
\usepackage{fancyhdr}

\newcounter{protocol}

\usepackage{xcolor}

\definecolor{moss}{RGB}{44,186,0}

\usepackage{orcidlink}

\begin{document}

\title{Scalable Qumode-Qubit State Transfer and Fast-forward Quantum Fourier Transform using Oscillators}
\author{Joel Bierman}
\affiliation{Department of Electrical and Computer Engineering, North Carolina State University, Raleigh, NC 27606, USA}
\author{Shubdeep Mohapatra}
\affiliation{Department of Electrical and Computer Engineering, North Carolina State University, Raleigh, NC 27606, USA}
\author{Huiyang Zhou}
\affiliation{Department of Electrical and Computer Engineering, North Carolina State University, Raleigh, NC 27606, USA}
\author{Yuan Liu}%
 \email{q\_yuanliu@ncsu.edu}
\affiliation{Department of Electrical and Computer Engineering, North Carolina State University, Raleigh, NC 27606, USA}
\affiliation{Department of Computer Science, North Carolina State University, Raleigh, NC 27606, USA}
\affiliation{Department of Physics and Astronomy, North Carolina State University, Raleigh, NC 27607, USA}
\date{\today}

\begin{abstract}
Transferring the information stored in the expansion coefficients of a multi-qubit state to the coefficients of a continuous-variable state is an important protocol for communicating quantum information. It was shown in previous work how to transfer an $n$-qubit state to a single qumode in $\mathcal{O}(2^n)$ time. We show that by transferring this state to $m$ qumodes, the runtime can be improved to $\mathcal{O}(2^{n/m})$. Furthermore, we demonstrate how multi-qumode state transfer can be used as a subroutine for approximately realizing the $n$-qubits quantum Fourier transform on $m$-qumode with runtime scaling $\mathcal{O}(m2^{n/m}/\epsilon+m^2)$, accelerating qubit quantum Fourier transform using qumodes. This work presents a scalable approach to convert discrete and continuous quantum information between an arbitrary number of qubits and qumodes. It represents a crucial step forward in mixed analog-digital quantum signal processing for computing, sensing, and communication.
\end{abstract}

\maketitle

\section{Introduction}

Quantum computing using bosonic modes (continuous variable or CV) as information carriers has the potential to speed up the runtime of many computational primitives beyond that of qubits (discrete variable or DV).\cite{liu2024hybrid,kemper2025hybrid, nourse2025using, braunstein_quantum_2005, navickas_experimental_2025, ha_analog_2025} The intuition behind this is that a single qumode can in principle store wavefunctions in Hilbert spaces of arbitrarily large (subject to physical implementation constraints) Hilbert spaces, whereas a single qubit can only store wavefunctions of a two-level system.

For example, the encoding of a logical qubit (or qudit) as a two-dimensional ($d$-dimensional) subspace of a high-dimensional Hilbert space is essential to quantum error correction.\cite{raussendorf_fault-tolerant_2007, fowler_high-threshold_2009, fowler_surface_2012, girvin_introduction_2023, horsman_surface_2012, bravyi_high-threshold_2024, litinski_game_2019} This can be accomplished using a single bosonic oscillator in place of a large number of physical qubits. Examples include GKP codes\cite{gottesman_encoding_2001, lin_closest_2023, hastrup_analysis_2023, grimsmo_quantum_2021, brady_advances_2024} and cat codes~\cite{albert_pair-cat_2019, cochrane_macroscopically_1999, lee_fault-tolerant_2024, albert_bosonic_2022}. Notably, the use of GKP codes beyond break even has been demonstrated for qubits~\cite{sivak_real-time_2023} as well as level 3 and 4 qudits~\cite{brock_quantum_2025}.

The simulation of physical systems whose Hilbert spaces are inherently high-dimensional (\emph{e.g.} vibronic spectra of molecules) is another use-case where CV information carriers may have an advantage over qubits.~\cite{kemper2025hybrid, nourse2025using, ha_analog_2025, navickas_experimental_2025} For example, compiling native bosonic operations to qubit gates is often extremely costly.~\cite{liu2024hybrid} The quantum Fourier transform (QFT)~\cite{Nielsen_Chuang_2010, nam_approximate_2020, kahanamoku2025log, coppersmith2002approximate} is another example of an algorithm where the use of high-dimensional information carriers may enable a speed-up over qubits. The exact versions of both qubit and qudit QFT~\cite{pavlidis_quantum-fourier-transform-based_2021} scale with the number of information carriers $n$ as $\mathcal{O}(n^2)$. Thus, if one encodes a $d$-level qudit into an oscillator, then a $\mathcal{O}(\mathrm{log}_2^2(d))$ speed-up over qubits can (in principle) be achieved.

In addition to quantum error correction and computation, more generally, today's wireless and communication infrastructure uses electromagnetic waves (bosonic modes) as information carrier. Unleashing the quantum nature of bosonic modes for storage and manipulating information will open new possibilities for quantum communication systems.~\cite{obrien_photonic_2009, luo_recent_2023, wang_scalable_2025, wu_continuous-variable_2022} For example, consider the possibility that early fault-tolerant quantum computers may not be monolithic (\emph{i.e.} all qubits occupy a single chip), but could consist of of logical qubits distributed across many nodes.~\cite{Arquin} Thus, we will need efficient fault-tolerant protocols not just for the storage and processing of high-dimensional quantum states, but their transportation across extended distances as well. This can be accomplished by using ion shuttling on trapped ion platforms, for example.~\cite{akhtar_high-fidelity_2023} Photonic interconnects~\cite{stute_quantum-state_2013, main_distributed_2025}, wherein one transports states between distant qubits using photons as an intermediary, is another means of quantum communication across computational nodes. Teleportation methods~\cite{Eisert2000, stack2025assessing}, wherein shared Bell pair states are consumed as a resource to enable the teleportation of states and nonlocal operations, is another means of carrying out distributed quantum computation. Teleportation methods require the consumption of many low-fidelity Bell states to distill a single high-fidelity Bell state, which adds to the qubit cost of the computation~\cite{qLDPC_bell_state_distillation}. Ion shuttling is a promising potential tool for realizing distributed quantum computing, but is limited in the sense that physically moving ions will destroy any quantum information stored in the CV modes when the ion chain is disrupted. Thus, moving a single ion can only transport one qubit of quantum information at a time.

Generalizing previous work on transferring the state of an arbitrary number of qubits to a single qumode~\cite{liu2024mixed} to a protocol where we can use an arbitrary number of qumodes is a necessary step towards making this framework scalable. To see this, consider the fact that demonstrations of quantum advantage will almost certainly necessitate over 100 logical qubits, each of which is likely to be comprised of dozens to tens of thousands of physical qubits, depending on the error correction code used and the necessary logical error rate. The single-qumode state transfer protocol involves a maximum displacement (and thus energy consumption) that scales linearly with the Hilbert space dimension of the qubit subsystem being transferred. Displacing a qumode by an amount on the order of $\sim 2^{100}$ to $2^{100,000}$ (even in units of femtometers) is infeasible for all intents and purposes. Thus, if we want this protocol to be feasible to implement, the maximum displacement to which any single qumode is subjected has to be minimized.

In this work, we propose a protocol wherein an $n$-qubit state is transferred to $m$ qumodes. The method pairs each of the $m$ qumodes to $\frac{n}{m}$ qubits and consists of a tensor product of $m$ instances of the previously proposed single-qumode state transfer procedure~\cite{liu2024mixed}. Whereas the runtime of the single-qumode state transfer scaled as $\mathcal{O}(2^n)$, the runtime of the method in the current work scales as $\mathcal{O}(2^{n/m})$. That is, provided that we can keep the ratio $\frac{n}{m}$ fixed, we can transfer states of arbitrary dimension to and from qumodes in constant time. This result is detailed in Theorem~\ref{thm: abelian state transfer}. We show that not only does this \emph{exponentially} reduce the amount of energy required for the process, but we show in Theorem~\ref{thm: photon loss} that the photon loss rate is exponentially suppressed as well. We also demonstrate in Sec.~\ref{sec: error suppression} that this multi-qumode state transfer procedure can passively suppress some position displacement errors, which may be of independent interest and which we expect will have positive consequences for its use in distributed quantum computation.

Additionally, we show that this state transfer procedure can be used not only for communication but as an algorithmic subroutine for computation as well. We detail how this state transfer procedure can be used to approximately implement the quantum Fourier transform with runtime scaling $\mathcal{O}(m2^{n/m}/\epsilon + m^2)$, where $\epsilon$ is the leading order contribution to the error. The full result is given in Theorem~\ref{thm: QFT theorem}.

This paper is structured as follows. Sec.~\ref{sec: background} reviews necessary prerequisite information for the rest of the work, including background on quantum signal processing (QSP) and single mode state transfer. Sec.~\ref{sec: multi-mode protocols} presents our work on the multi-mode parallel state transfer and the multi-mode QFT. Sec.~\ref{sec: runtime and fidelity bounds} provides runtime and fidelity bound complexities for the state transfer and qumode QFT protocols. Sec.~\ref{sec: error suppression} shows how the abelian state transfer procedure can passively suppress some position displacement errors (assuming perfect implementation of the state transfer circuit itself). Sec.~\ref{sec: photon loss rate} details how using multiple qumodes for state transfer exponentially suppresses the photon loss rate compared to its single-qumode counterpart. Sec.~\ref{sec: state transfer numerical simulation} provides a numerical demonstration of the abelian multi-qumode state transfer. Sec.~\ref{sec: comparison to qubit qft} presents a back-of-the-envelope resource estimation comparison to qubit QFT and presents approximate conditions for when qumode QFT has favorable runtime scaling. In Sec.~\ref{sec: conclusion} we finish with a discussion, concluding remarks, and potential future directions of research.

\section{Background}\label{sec: background}

In this section we give a brief overview of the prerequisite conceptual ideas that are necessary for the rest of the paper. Sec.~\ref{sec: hybrid GQSP} outlines QSP on qubits on hybrid CV-DV systems. Sec.~\ref{sec: single-qumode state transfer} outlines the single-qumode state transfer protocol.

\subsection{Hybrid Single-Variable QSP}\label{sec: hybrid GQSP}

QSP~\cite{MRTC21, rossi_multivariable_2022, martyn_efficient_2023, low_optimal_2017, joven_scalable_2026, lin2025mathematical, gilyen_quantum_2019, lin_optimal_2020} is a method wherein one takes a block encoding of either scalar or matrix quantities and generates block encodings of polynomial transformations of those objects. If we can realize both conditional displacements:
\begin{equation}
    W^{(\kappa)}=\ket{0}\bra{0}\otimes e^{i\tfrac{\kappa}{2}\hat{x}} + \ket{1}\bra{1}\otimes e^{-i\tfrac{\kappa}{2}\hat{x}} = \begin{bmatrix}
        e^{-i\tfrac{\kappa}{2} \hat{x}} & 0\\
        0 & e^{i\tfrac{\kappa}{2}\hat{x}}
    \end{bmatrix}
\end{equation} and single-qubit Pauli-X rotations $e^{i\phi_j\sigma_x}$, then we can prepare block-encodings of degree $d$ polynomials of $e^{i\tfrac{\kappa}{2}\hat{x}}$ into the qubit subspace as \cite{liu2024mixed}:
\begin{equation}
    \begin{bmatrix}
        P(e^{-i\tfrac{\kappa}{2}\hat{x}}) & iQ(e^{-i\tfrac{\kappa}{2}\hat{x}})^\dagger\\
        iQ(e^{i\tfrac{\kappa}{2}\hat{x}}) & P(e^{i\tfrac{\kappa}{2}\hat{x}})^\dagger \\
    \end{bmatrix} = e^{i\phi_0\sigma_x}\prod_{j=1}^d W^{(\kappa)}e^{i\phi_j\sigma_x}.
\end{equation}Here, $P$ and $Q$ are degree $d$ Laurent polynomials:
\begin{align}
    P(e^{-i\tfrac{\kappa}{2}\hat{x}}) = \sum_{n=-d}^dp_ne^{-i\tfrac{n\kappa}{2}\hat{x}}\\
    Q(e^{-i\tfrac{\kappa}{2}\hat{x}}) = \sum_{n=-d}^dq_ne^{-i\tfrac{n\kappa}{2}\hat{x}}.
\end{align}The protocol is then to compute real-valued Fourier coefficients $p_n$ and $q_n$ whose expansions give approximations to target functions $p(\hat{x})$ and $q(\hat{x})$, respectively.
\subsection{ Abelian Single-Mode State Transfer}\label{sec: single-qumode state transfer}

In abelian single-mode state transfer~\cite{liu2024mixed}, we start with the $n$-qubit state:
\begin{align}
    \ket{\psi}_Q = \sum_{x=0}^{2^n-1}C_x\ket{x}_Q
\end{align}where $x = \sum_{j=1}^nx_n \cdot2^{n-1}$ and $\ket{x} = \bigotimes_{j=1}^n\ket{x_j}$. We wish to transfer the coefficients $C_x$ to be the coefficients of Gaussian wavepackets centered about position $x\Delta$ on a qumode. That is, we want to produce the corresponding qumode state:
\begin{align}
    \ket{\psi}_O = \sum_{x=0}^{2^n-1}C_x\ket{x,\Delta}^{gauss}_O
\end{align}
where:
\begin{align}
    \ket{x, \Delta}^{gauss}_O = \frac{1}{(2\pi\sigma^2)^{1/4}}\int dq e^{-(q-x\Delta)^2/4\sigma^2}\ket{q}_O
\end{align}where $\ket{q}$ is a position basis state of the oscillator. Here, $\Delta$ represents the spacing in our lattice of Gaussian wavepackets. If $\sigma \ll \Delta$, then these Gaussian wavepackets are orthogonal to a good approximation.

We initialize our hybrid qubit-oscillator state as:
\begin{align}
    \ket{\psi}_Q\otimes\ket{0,\Delta}^{gauss}_O.
\end{align}Note that each gate in the following sequence of controlled displacement gates:
\begin{align}
    \prod_{j=1}^n cD_j(2^{j-1}\Delta)
\end{align}displaces the qumode by an amount $2^{j-1}\Delta$ if and only if the $j$th bit of the qubit basis state $\ket{x}_Q$ is $1$. Therefore, the aggregate action of this gate sequence is, for any qubit basis state $\ket{x}_Q$, to displace the qumode by an amount $x\Delta$. This gives us the entangled state:
\begin{align}
    \sum_{x=0}^{2^n-1}C_x\ket{x}_Q\ket{x,\Delta}^{gauss}_O.
\end{align}We can then apply a sequence of $n$ CV-DV QSP sequences, where the $jth$ QSP sequence acts on the $jth$ qubit, performing a bit-flip if and only if the $jth$ bit in the binary expansion of $x$ in $\ket{x,\Delta}^{gauss}_O$ is 1 and acts as the identity otherwise. This can be realized by block encoding polynomial approximations to square wave functions as:
\begin{equation}
    R_j(\hat{x}) = \begin{pmatrix}
        S_j(\hat{x}) & \sqrt{1-S_j(\hat{x}})\\
        -\sqrt{1-S_j(\hat{x})} & S_j(\hat{x})
    \end{pmatrix}
\end{equation}where:
\begin{equation}\label{eq: original square wave}
    S_j(\hat{x})= \Theta\left(\mathrm{cos}\left[\tfrac{\pi}{2^{n-j}}(\tfrac{\hat{x}}{\Delta}-2^{n-j-1}+\tfrac{1}{2})\right]\right)
\end{equation}This will disentangle the qubits from the oscillator, giving us the state:
\begin{align}
    &\ket{0}^{\otimes n}_Q\otimes\sum_{x=0}^{2^n-1}C_x\ket{x,\Delta}_O^g\\
    &=\ket{0}_Q^{\otimes n}\otimes\ket{\psi}_O.
\end{align}The circuit for this protocol is given in Fig.~\ref{fig:single-mode state transfer}.
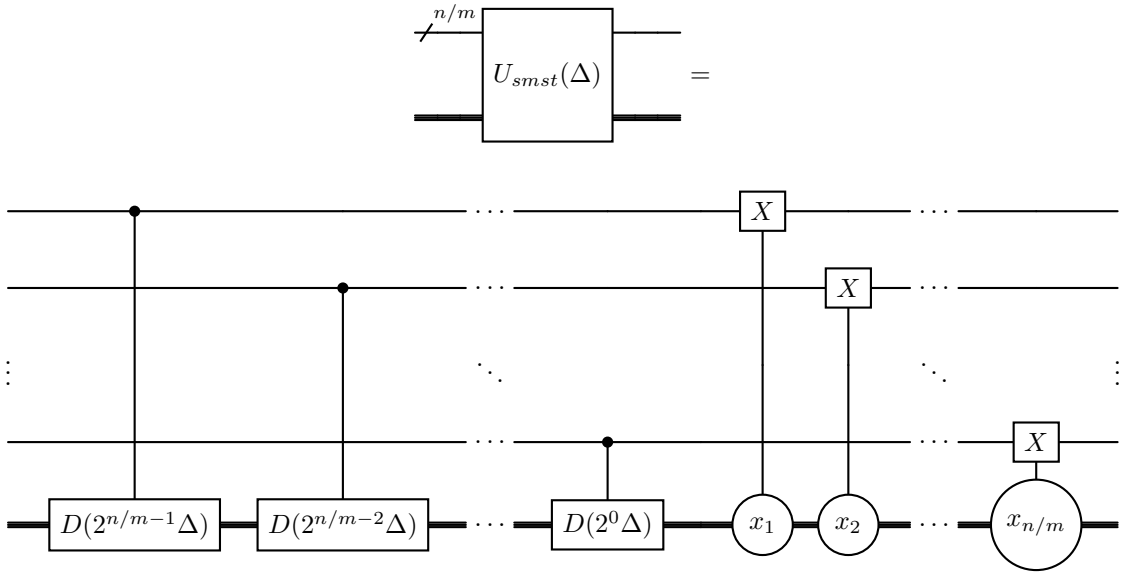
\begin{figure*}[ht!]
    \centering
    \begin{quantikz}[wire types={q,b}, column sep = 0.3cm]
            & \qwbundle[]{n/m} & & \gate[2]{U_{smst}(\Delta)}& & &\\
            & & & & & &\\
        \end{quantikz} = 
        \begin{quantikz}[wire types={q,q,n,q,b}]
            & \ctrl{4} & & \ \ldots \ & & & \gate[]{X} & & \ \ldots \ & &\\
            & & \ctrl{3} & \ \ldots \ & & & \wire[u]{q} & \gate[]{X} & \ \ldots \ & &\\
            \vdots & & & \ddots & & & \wire[u]{q} & \wire[u]{q} & \ddots & & \vdots\\
            & & & \ \ldots \ & \ctrl{1} & & \wire[u]{q} & \wire[u]{q} & \ \ldots \ & \gate[]{X} &\\
            & \gate[]{D(2^{n/m-1}\Delta)} & \gate[]{D(2^{n/m-2}\Delta)} & \ \ldots \ & \gate[]{D(2^0\Delta)} & & \gate[style=circle]{x_1}\wire[u]{q} & \gate[style=circle]{x_2}\wire[u]{q} & \ \ldots \ & \gate[style=circle]{x_{n/m}}\wire[u]{q} &\\
        \end{quantikz}
    \caption{The single-mode state transfer circuit. The first stage consists of a sequence of controlled displacement gates whose net action is, for any qubit basis state $\ket{x}$, to displace the qumode by an amount $x\Delta$. In the second stage, a sequence of QSP sequences, where the $j$th sequence flipps the $j$th bit conditioned on the $j$th bit of $x$ being 1.}
    \label{fig:single-mode state transfer}
\end{figure*}
\subsection{Non-abelian single-mode state transfer}

While we primarily focus on the abelian state transfer procedure in this work, it is also worth taking note that Ref.~\cite{liu2024mixed} also presents formulation of the state transfer protocol of Ref.~\cite{Hastrup2022} in terms of multi-variable non-abelian QSP. The single-mode version of this procedure applies the gate sequence:
\begin{equation}
\begin{split}
    U_{smst}^{na}(\Delta;n) = &e^{-i\tfrac{\Delta}{2}2^{n-1}\hat{p}\hat{\sigma}_x^{(j)}}e^{i\tfrac{\pi}{2^{n}\Delta}\hat{x}\hat{\sigma}_y^{(j)}}\\
    &\times\prod_{j=n-1}^1e^{i\tfrac{\Delta}{2}2^{j-1}\hat{p}\hat{\sigma}_x^{(j)}}e^{i\tfrac{\pi}{2^j\Delta}\hat{x}\hat{\sigma}_y^{(j)}}
\end{split}
\end{equation}to an initial state of the form:
\begin{equation}
\begin{split}
    \sum_{\textbf{s}}c_{\textbf{s}}\ket{\phi_{\textbf{s}}}_Q\otimes\ket{0,\Delta}_{O}^{\text{sinc}}
\end{split}
\end{equation}where:
\begin{equation}
    \ket{0,\Delta}_{O}^{\text{sinc}} = \frac{1}{\sqrt{\Delta}}\int dq~\mathrm{sinc}(\tfrac{\pi q}{\Delta})\ket{q}_O
\end{equation}and:
\begin{equation}
    \ket{\phi_{\textbf{s}}}_Q = (-1)^{\gamma_{\textbf{s}}}\cdot\bigotimes_{j=1}^n\left[\frac{1}{\sqrt{2}}(\ket{0}+\ket{1})_{Q_j}\right].
\end{equation}$\textbf{s}$ is a vector in the set $\{-1,1\}^n$. $\gamma_{\textbf{s}}$ is defined as:
\begin{equation}
    \gamma_{\textbf{s}} = \sum_{j=1}^{n-2}\frac{1}{2}(s_j+s_{j+1}) + \frac{1}{2}(s_{n-1}-s_n).
\end{equation}This will produce the state:
\begin{equation}
\begin{split}
    \ket{0}^{\otimes n}_Q\otimes\sum_{\textbf{s}}c_{\textbf{s}}\frac{1}{\sqrt{\Delta}}\int dq~\mathrm{sinc}\left(\frac{\pi(q-q_{\textbf{s}})}{\Delta}\right)\ket{q}_O
\end{split}
\end{equation}where:
\begin{equation}
    q_{\textbf{s}} = \frac{\Delta}{2}\left(\sum_{j=1}^{n-1}s_j2^{j-1}-s_n2^{n-1}\right).
\end{equation}The action of this procedure is directly analogous to that of the abelian version, except that the qumode wavefunction is expressed in a basis of sinc wavepackets. This procedure has runtime $\mathcal{O}(2^n)$ with fidelity:
\begin{equation}
\begin{split}
    1 - \mathcal{O}\left(\int_{-\infty}^{-\tfrac{\Delta}{2}(2^n-1)}dq|\psi(q)|^2 + \int_{\tfrac{\Delta}{2}(2^n-1)}^{\infty}dq|\psi(q)|^2\right)
\end{split}
\end{equation}where:
\begin{equation}
    \psi(q) = \frac{1}{\sqrt{\Delta}}\sum_{\textbf{s}}c_{\textbf{s}}\mathrm{sinc}\left(\frac{\pi(q-q_{\textbf{s}})}{\Delta}\right).
\end{equation}
\section{Multi-Mode Protocols}\label{sec: multi-mode protocols}

In this section, present the multi-qumode state transfer process and give a high level overview of the multi-qumode QFT. The details of the multi-qumode QFT procedure are deferred to Appendix.~\ref{sec: appendix QFT derivation}. In Sec.~\ref{sec: multimode state transfer} we present the multi-mode state transfer. In Sec.~\ref{sec: main body QFT} we present the construction of the multimode QFT.

\subsection{Multi-Mode State Transfer}\label{sec: multimode state transfer}

In the non-abelian single-qumode state transfer protocol, the qubit computational basis state $\ket{x}$ was encoded into a Gaussian wavepacket centered about position $x\Delta$. This was done using a sequence of controlled displacements between each of the qubits and the oscillator that would displace the oscillator by an amount corresponding to the binary expansion of $x$. That is, the oscillator functioned essentially as one $2^n$-dimensional qudit. 

In what follows, we formulate a multi-qumode state transfer procedure for the non-abelian case, but we note that all of the presented reasoning also applies to a generalization of the non-abelian single-mode state transfer $U_{smst}^{na}(\Delta;n)$. The construction is directly analogous.

The key insight in the generalization to the multi-qumode state transfer is that by partitioning the qubits into subsets, we can perform a lower-dimensional version of the single-qumode state transfer procedure between each subset and an oscillator. Thus, each oscillator functions as a low-dimensional qudit whose dimension we can control simply by changing the number of qubits in each subset. For $m$ oscillators, the maximum amount of displacement we need is $2^{\tfrac{n}{m}-1}\Delta$. The ratio $\frac{n}{m}$ is a quantity which we control, thus for arbitrary $n$, we can convert the exponential scaling of the single-qumode state transfer procedure to a (possibly large) constant overhead by increasing $m$ and keeping the ratio $\frac{n}{m}$ fixed. We now prove that the correctness of the multi-qumode state transfer procedure given in Fig.~\ref{fig:multimode state transfer} follows from the correctness of the single qumode state transfer protocol.

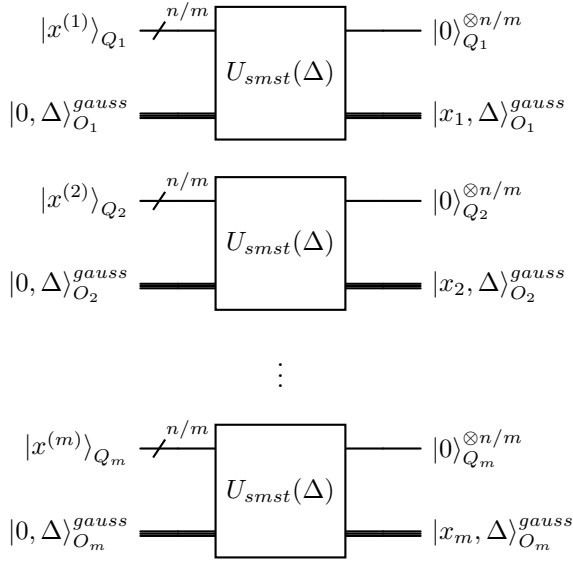
\begin{figure}[ht!]
    \centering
    \begin{quantikz}[wire types={q,b,q,b,n,q,b}]
        \lstick{$\ket{x^{(1)}}_{Q_1}$} & \qwbundle[]{n/m} & \gate[2]{U_{smst}(\Delta)} & & \rstick{$\ket{0}^{\otimes n/m}_{Q_1}$}\\
        \lstick{$\ket{0,\Delta}^{gauss}_{O_1}$} & & & & \rstick{$\ket{x_1,\Delta}_{O_1}^{gauss}$}\\
        \lstick{$\ket{x^{(2)}}_{Q_2}$} & \qwbundle[]{n/m} & \gate[2]{U_{smst}(\Delta)} & & \rstick{$\ket{0}^{\otimes n/m}_{Q_2}$}\\
        \lstick{$\ket{0,\Delta}_{O_2}^{gauss}$} & & & & \rstick{$\ket{x_2,\Delta}^{gauss}_{O_2}$}\\
        & & \vdots & & \\
        \lstick{$\ket{x^{(m)}}_{Q_m}$} & \qwbundle[]{n/m} & \gate[2]{U_{smst}(\Delta)} & &\rstick{$\ket{0}^{\otimes n/m}_{Q_m}$}\\
         \lstick{$\ket{0,\Delta}_{O_m}^{gauss}$} & & & & \rstick{$\ket{x_m,\Delta}^{gauss}_{O_m}$}
        \end{quantikz}
    \caption{The multi-mode state transfer and its action on a multi-qubit basis state. Here, the single-mode state transfer is denoted as $U_{smst}$.}
    \label{fig:multimode state transfer}
\end{figure}

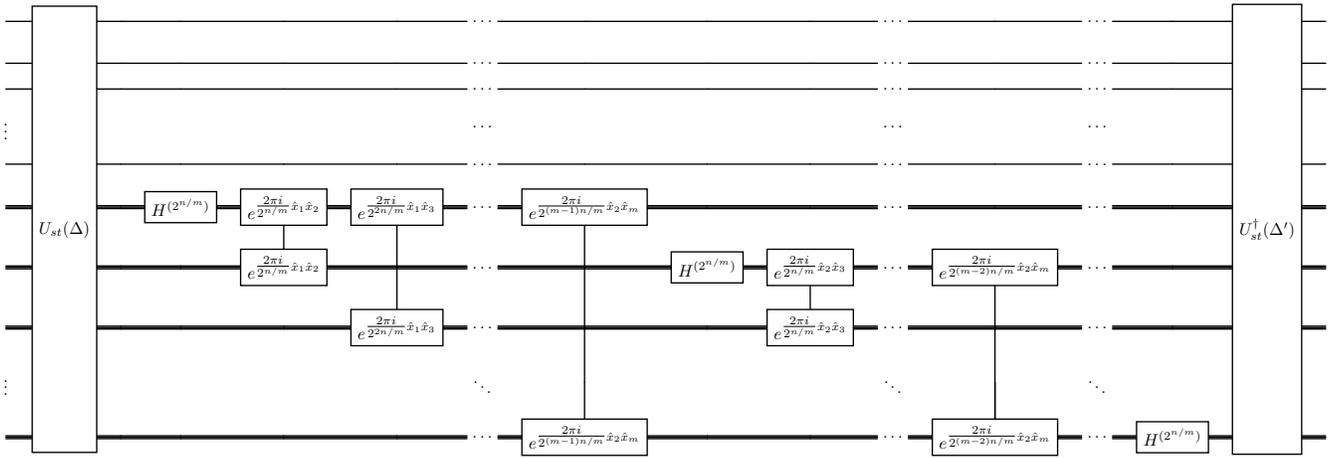
\begin{figure*}[ht]
    \centering
    \resizebox{\textwidth}{!}{%
    \begin{tikzpicture}
    \node[scale=0.65] {
    \begin{quantikz}[wire types={q,q,q,n,q,b,b,b,n,b}]
        & \gate[10]{U_{st}(\Delta)} & & & & & \ \ldots \ & & & & \ \ldots \ & & \ \ldots \ & & \gate[10]{U^\dagger_{st}(\Delta^\prime)} &\\
        & & & & & & \ \ldots \ & & & & \ \ldots \ & & \ \ldots \ & & &\\
        & & & & & & \ \ldots \ & & & & \ \ldots \ & & \ \ldots \ & & &\\
        \vdots & & & & & & \ \ldots \ & & & & \ \ldots \ & & \ \ldots \ & & &\\
        & & & & & & \ \ldots \ & & & & \ \ldots \ & & \ \ldots \ & & &\\
        & & & \gate[]{H^{(2^{n/m})}} & \gate[]{e^{\tfrac{2\pi i}{2^{n/m}}\hat{x}_1\hat{x}_2}}& \gate[]{e^{\tfrac{2\pi i}{2^{2{n/m}}}\hat{x}_1\hat{x}_3}} & \ \ldots \ & \gate[]{e^{\tfrac{2\pi i}{2^{(m-1){n/m}}}\hat{x}_2\hat{x}_m}} & & & \ \ldots \ & & \ \ldots \ & & &\\
        & & & & \gate[]{e^{\tfrac{2\pi i}{2^{n/m}}\hat{x}_1\hat{x}_2}}\wire[u]{q} & \wire[u]{q} & \ \ldots \ & \wire[u]{q} &\gate[]{H^{(2^{n/m})}} & \gate[]{e^{\tfrac{2\pi i}{2^{n/m}}\hat{x}_2\hat{x}_3}} & \ \ldots \ & \gate[]{e^{\tfrac{2\pi i}{2^{(m-2){n/m}}}\hat{x}_2\hat{x}_m}} & \ \ldots \ & & &\\
        & & & & & \gate[]{e^{\tfrac{2\pi i}{2^{2{n/m}}}\hat{x}_1\hat{x}_3}}\wire[u]{q} & \ \ldots \ & \wire[u]{q} & & \gate[]{e^{\tfrac{2\pi i}{2^{n/m}}\hat{x}_2\hat{x}_3}}\wire[u]{q} & \ \ldots \ & \wire[u]{q} & \ \ldots \ & & &\\
        \vdots & & & & & & \ddots & \wire[u]{q} & & & \ddots & \wire[u]{q} & \ddots & &\\
        & & & & & & \ \ldots \ & \gate[]{e^{\tfrac{2\pi i}{2^{(m-1){n/m}}}\hat{x}_2\hat{x}_m}}\wire[u]{q} & & & \ \ldots \ & \gate[]{e^{\tfrac{2\pi i}{2^{(m-2){n/m}}}\hat{x}_2\hat{x}_m}}\wire[u]{q}\wire[u]{q}& \ \ldots \ & \gate[]{H^{(2^{n/m})}}& &\\
    \end{quantikz}
    };
    \end{tikzpicture}
    }
    \caption{The multi-qumode QFT circuit. Ancilla qubits used for the qumode Hadamard gates are omitted for simplicitly.      }
    \label{fig:multimode QFT circuit}
\end{figure*}

Consider the simplified case where the number of qubits in each partition is the same. Index each partition of qubits by $j$. The state that we wish to transfer is given by:
\begin{equation}
    \ket{\psi}_Q = \sum_{x^{(1)},\dots,x^{(m)}=0}^{2^{n/m}-1}C_{x^{(1)}\cdots x^{(m)}}\bigotimes_{j=1}^m\ket{x^{(j)}}
\end{equation}We want a circuit $U_{st}$ such that:
\begin{equation}\label{eq: multi-qumode state transfer}
\begin{split}
    U_{st}&\ket{\psi}_Q\bigotimes_{j=1}^m\ket{0, \Delta}_{O_j}^{gauss}=\ket{0}^{\otimes n}\otimes \ket{\psi}_O\\
\end{split}
\end{equation}where
\begin{equation}
    \ket{\psi}_O = \sum_{x^{(1)},\dots,x^{(m)}=0}^{2^{n/m}-1}C_{x^{(1)}\cdots x^{(m)}}\bigotimes_{j=1}^m\ket{x^{(j)}, \Delta}_{O_j}^{gauss}
\end{equation}The fact that we can accomplish this using $m$ copies of the single-qumode state transfer circuit acting in parallel on each qubit partition follows from the correctness of the single-qumode circuit and the linearity of quantum gates. Throughout this work, we will use the subscript $smst$ (``single-mode state transfer") to distinguish it from the multi-mode state transfer process as a whole. Denote a single-qumode state transfer circuit acting between on qubit partition and oscillator $j$ by $U_{smst}^{(j)}$. The action of $U_{smst}^{(j)}$ on $\ket{x^{(j)}}_{Q_j}\otimes\ket{0,\Delta}_{O_j}^{gauss}$ is given by:
\begin{equation}
    U_{smst}^{(j)}\ket{x^{(j)}}_{Q_j}\otimes\ket{0,\Delta}_{O_j}^{gauss} = \ket{0}^{\otimes m}\otimes\ket{x^{(j)},\Delta}_{O_j}^{gauss}
\end{equation}therefore:
\begin{equation}
    \bigotimes_{j=1}^mU_{smst}^{(j)}\ket{x^{(j)}}\otimes\ket{0,\Delta}_{O_j}^{gauss} = \bigotimes_{j=1}^m\ket{0}_{Q_j}\otimes\ket{x^{(j)},\Delta}_{O_j}^{gauss}.
\end{equation}Thus, by the linearity of quantum circuits, it follows that a tensor product of such subcircuits will yield the transformation in Eq.~\eqref{eq: multi-qumode state transfer}. Thus, the oscillator basis state $\ket{x^{(1)},\Delta}^{gauss}_{O_1}\otimes\cdots\otimes\ket{x^{(m)},\Delta}^{gauss}_{O_j}$ encodes the integer $x$ whose d-ary expansion is $x = x^{(1)}\cdot d^{0}+x^{(2)}\cdot d^{1}+\cdots+x^{(m)}\cdot d^{m-1}$, where $d = 2^{n/m}$.

\subsection{Multi-Mode QFT}\label{sec: main body QFT}

One use-case of the multi-qumode state transfer protocol is that we can implement subroutines with qubit logic inputs and outputs by offloading them to qumodes. We transfer the state from the qubits to the qumodes, implement a sequence of gates which utilize the high-dimensional nature of the qumodes, then transfer the state back to qubits.

The quantum Fourier transform (QFT), well-known for its use in other quantum algorithms such as phase estimation or Shor's algorithm, is one such candidate for this strategy. It is defined as the following operator acting on a Hilbert space of dimension $N$:
\begin{equation}
\begin{split}
    QFT_N = \frac{1}{\sqrt{N}}\sum_{x,y=0}^{N-1}e^{2\pi ixy/N}\ket{x}\bra{y}.
\end{split}
\end{equation}
Running the QFT on both qubits and $d$-level qudits scales quadratically with the number of information carriers. However, for a fixed number $n$ of information carriers, $d$-level qudits will have a Hilbert space dimension of $d^N$ as opposed to $2^N$ for qubits. The number of qubits needed to simulate $N$ qudits scales as $N\mathrm{log}(d)$, therefore in principle qudit QFT can achieve a polylogarithmic runtime speedup in the qudit level over qubit QFT.

The intuition behind our current approach is that if we can construct a set of qumode gates whose logical action on the qumodes is the same as the corresponding qudit gates in the qudit QFT, then it follows that we can construct the entire ``logical" qudit QFT on the qumodes. We obtain the QFT on the qubits by conjugating this gate sequence with a state transfer procedure and its inverse. Since, for $\tfrac{n}{m}$ qubits per qumode in the state transfer, each qumode acts as a $2^{n/m}$-level system, our goal is to implement primitives which enable $2^{n/m}$-level qudit logic.

Qudit QFT involves two primitives: a qudit Hadamard and a qudit controlled rotation gate.~\cite{pavlidis_quantum-fourier-transform-based_2021}. The actions of each of these gates on qudit basis states $\ket{x^{(1)}}$ and $\ket{x^{(1)}}\otimes\ket{x^{(2)}}$ are given by:
\begin{equation}\label{eq: qudit Hadamard gate}
    H^{(d)}\ket{x^{(j)}} = \frac{1}{\sqrt{d}}\sum_{x^{(j)}=0}^{d-1}\sum_{y^{(j)}=0}^{d-1}e^{i\tfrac{2\pi}{d}x^{(j)}y^{(j)}}\ket{y^{(j)}}
\end{equation}and:
\begin{equation}
    R_{k}^{(d)}\ket{x^{(j)}}\ket{x^{(j^\prime)}} = e^{i\tfrac{2\pi}{d^k}x^{(j)}x^{(j^\prime)}}\ket{x^{(j)}}\ket{x^{(j^\prime)}}.
\end{equation}Thus, what we need are  qumode gate sequences that implement the same logic on our Gaussian wavepacket states $\ket{x^{(j)},\Delta}^{gauss}_{O_j}$. We can then apply these gates in the same fashion as the qudit QFT gate. We emphasize that although the qumodes in our construction bear resemblance to qudits in that we are using them to store discrete $d$-level information, they are functionally not qudits in the strictest sense of the word. This is because qudit systems obey \emph{modulo} $d$ logic, whereas qumodes do not. \emph{i.e.} shift operations on Gaussian wavepackets are non-periodic.

\subsection{Mapping qudit gates to oscillator gates}

In this section, we outline the construction of the Hadamard and controlled rotation gates on qumodes. We then will put these building blocks together to implement the QFT circuit on qumodes. This amounts to implementing logical qudit gates which obey modulo $d$ arithmetic on non-periodic oscillator systems which do not. This non-periodicity presents no issues when implementing diagonal qudit operators such as the controlled rotation gate, however we will soon see how this presents significant challenges for efficiently implementing non-diagonal operators such as the qudit Hadamard gate.

\subsubsection{Hadamard gate}

We first note that the action of the qudit Hadamard in Eq.\eqref{eq: qudit Hadamard gate} is the same as the quantum Fourier transform on $\mathrm{log}_2(d)$ qubits. Implementing the qubit quantum Fourier transform on qubits using a single qumode as an ancilla has already been worked out in Ref.~\cite{liu2024mixed} We can take inspiration from this protocol to construct the qudit Hadamard gate, however we note two main key differences between the goals of that gate sequence and those of the current work.

The first is that the exact action of the QFT is never precisely implemented on the qumodes themselves in the previous work. It only emerges on the qubits after completing the inverse state transfer, with the qumode being left in a linear combination of Gaussian amplitudes that encode information we do not need. In the present work, we need the opposite, with the qumode containing the Fourier coefficients and a collection of ancilla qubits containing discardable information after some inverse ``state transfer-like" gate sequence.

The second difference is that in the single-qumode protocol, the state being initially transferred is already periodic through the use of $a$ padding qubits initialized in the $\ket{+}^{\otimes a}$ state. In the present work, we need a protocol which first transfers a non-periodic state to the qumodes, then makes the state of each qumode periodic mid-circuit. The reason for this is the entangling operations required by the qudit QFT sequence will induce erroneous phases if the position of any of the oscillators contains redundant information about the periodicity in addition to the data itself.

In this section, we outline the construction of the qudit Hadamard gate. The details of the derivation are in Appendix.~\ref{sec: Hadamard gate derivation}. We start by supposing that we have already transferred the state from $\tfrac{n}{m}$ qubits to the oscillator such that the qumode state is a linear combination of Gaussian wavepackets $\ket{x^{(j)}, \Delta}^{gauss}_{O_j}$. The complete gate sequence can be enumerated as follows:

\begin{enumerate}
    \item Apply a ``padding" gate $U_p^{(j)}(\Delta)$ between an $a$-qubit ancilla register and the $j$th qumode.
    \item Apply a displacement gate $D_j(-\tfrac{2^{n/m+a}}{2}\Delta)$ to the $j$th qumode.
    \item Apply the free evolution gate $F_j$ to the $j$-th qumode.
    \item Apply an ''anti-padding'' gate $U_{ap}^{(j)}(\Delta^\prime)$ between the $j$th qumode and $a^\prime$ anti-padding qubits. 
\end{enumerate}

The purpose of the first step is to make the wavefunction periodic with respect to $x^{(j)}$ so that the non-periodic qumode behaves logically like a qudit. That is, Ref.~\cite{liu2024mixed} found that the free evolution gates gives rise to the Fourier phases with error $\mathcal{O}(1/2^a)$ when padded with $a$ qubits. The same is true for step 3 in the above procedure. We can accomplish this by requiring that the padding gate has the following action:
\begin{equation}
\begin{split}
    &U_p^{(j)}\ket{+}^{\otimes a}\ket{x^{(j)},\Delta}_{O_j}^{gauss}\\
    &= \ket{0}^{\otimes a}\frac{1}{\sqrt{2^a}}\sum_{k^{(j)}=0}^{2^a-1}\ket{2^{n/m}k^{(j)}+x^{(j)}, \Delta}^{gauss}_{O_j}.
\end{split}
\end{equation}Note the similarity with the action of the state transfer circuit. The key difference here is that $2^{n/m}k^{(j)}+x^{(j)}$ is an integer where bits $a+1,\dots,a+\tfrac{n}{m}$ encode information about $x^{(j)}$ and bits $\tfrac{n}{m}+1,\dots\tfrac{n}{m}+a$ encode information about the periodicity index $k^{(j)}$. We can accomplish this with a version of the single-variable state transfer sequence that has been modified in the following two ways. The first is that the controlled displacements in the entangling stage use an effective spacing $2^{n/m}\Delta$. The second is that the quare wave polynomials in the QSP sequences of the disentangling stage act only on bits $1,\dots,a$ of the integer encoding the position of the qumode. The circuit for this padding gate is given in Fig.~\ref{fig:padding gate}. Here, the disentangling stage consists of $a$ QSP sequences which block encode a polynomial approximation of the modified square wave:
\begin{equation}\label{eq: forward padding GQSP Polynomial}
    S_{j,pad}^{(\Delta,a)}(\hat{x})= \Theta\left(\mathrm{cos}\left[\tfrac{\pi}{2^{{n/m}+a-j}}(\tfrac{\hat{x}}{\Delta}-2^{{n/m}+a-j-1}+\tfrac{1}{2})\right]\right)
\end{equation}for $j=1,\dots,a$. We note that this has been modified from Eq.~\eqref{eq: original square wave} in that the number of data qubits $n$ has been replaced by $\tfrac{n}{m} + a$ and $j$ runs from $1$ to $a$ instead of $1$ to $n$.

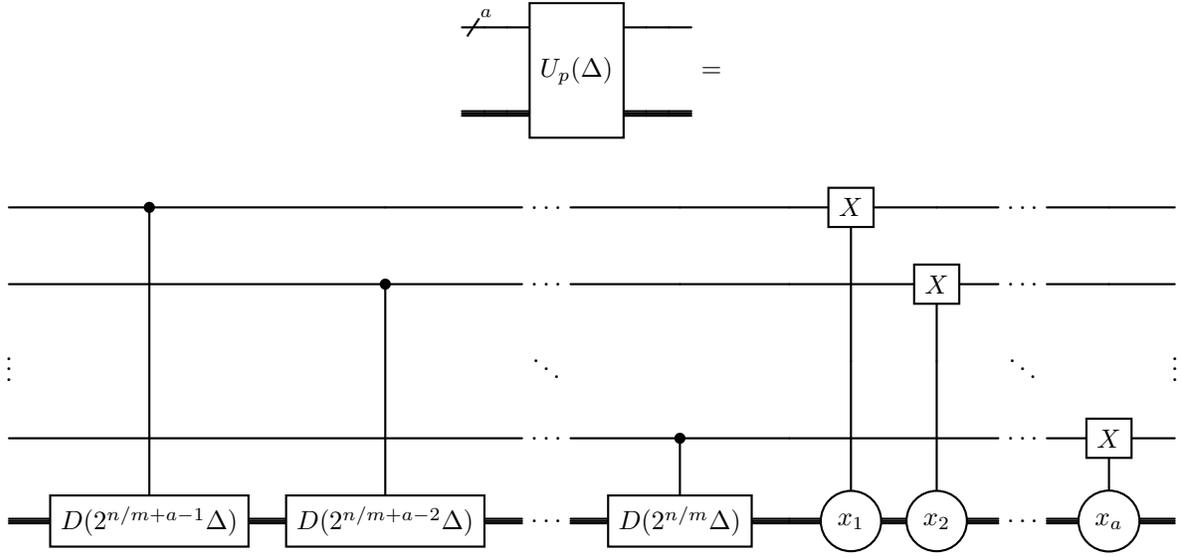
\begin{figure*}[ht!]
    \centering
    \begin{quantikz}[wire types={q,b}, column sep = 0.3cm]
            & \qwbundle[]{a} & & \gate[2]{U_{p}(\Delta)}& & &\\
            & & & & & &\\
        \end{quantikz} = 
        \begin{quantikz}[wire types={q,q,n,q,b}]
            & \ctrl{4} & & \ \ldots \ & & & \gate[]{X} & & \ \ldots \ & &\\
            & & \ctrl{3} & \ \ldots \ & & & \wire[u]{q} & \gate[]{X} & \ \ldots \ & &\\
            \vdots & & & \ddots & & & \wire[u]{q} & \wire[u]{q} & \ddots & & \vdots\\
            & & & \ \ldots \ & \ctrl{1} & & \wire[u]{q} & \wire[u]{q} & \ \ldots \ & \gate[]{X} &\\
            & \gate[]{D(2^{n/m+a-1}\Delta)} & \gate[]{D(2^{n/m+a-2}\Delta)} & \ \ldots \ & \gate[]{D(2^{n/m}\Delta)} & & \gate[style=circle]{x_{1}}\wire[u]{q} & \gate[style=circle]{x_{2}}\wire[u]{q} & \ \ldots \ & \gate[style=circle]{x_{a}}\wire[u]{q} &\\
        \end{quantikz}
    \caption{The padding gate. The first stage consists of a sequence of controlled displacement gates which, for any $a$-qubit basis state $\ket{k}$ displaces the qumode by an amount $2^{n/m}k\Delta$. The second stage consists of a sequence of disentangling QSP sequences which flip the $j$th qubit conditioned on the $j$th bit of the integer $k$ encoded in the oscillator being 1.}
    \label{fig:padding gate}
\end{figure*}

The derivation of the action and error bounds of steps 2 and 3 follow directly from the work of Ref.~\cite{liu2024mixed}. That is, steps 2 and 3 give us (approximately to error $\mathcal{O}(1/2^a) + \mathcal{O}(\tfrac{\sigma}{\Delta})$) the state:
\begin{equation}\label{eq: post Fourier gate H main body}
\begin{split}
    &\left(\frac{2\pi\sigma}{\Delta 2^{n/m}}\sqrt{\frac{2}{\pi}}\right)^{1/2}\sum_{\tilde{k}^{(j)}={-\infty}}^{\infty}e^{2\pi i y^{(j)}x^{(j)}/2^{n/m}}\\
    &\times e^{-\left(\tfrac{2\pi}{\Delta}\tilde{k}^{(j)}\right)^2\sigma^2}\ket{2^{n/m}\tilde{k}^{(j)}+y^{(j)}, \Delta^\prime}^{gauss}_{O_j}.
\end{split}
\end{equation}Naively, one might think that all we need to do is apply the inverse of the padding gate with the new lattice spacing $\Delta^\prime$, however this turns out to not be the case for two reasons. The first is that the new periodicity index $\tilde{k}^{(j)}$ runs over negative indices as well as positive indices. If our goal is to transfer information about $\tilde{k}^{(j)}$ to the qubits, then we need one ancilla qubit that stores the sign of $\tilde{k}^{(j)}$, otherwise the operation will not be reversible. The second reason is that our qumode wavefunction is now much more spread out in the position quadrature due to the $e^{-\left(\tfrac{2\pi}{\Delta}\tilde{k}^{(j)}\right)^2\sigma^2}$ Gaussian envelope terms. The implication of these Gaussian envelopes is that we need a number of anti-padding qubits $a^\prime$ which depends on the ratio $\tfrac{\Delta}{\sigma}$ so that we may truncate the summation over $\tilde{k}^{(j)}$ to finite order with low error.

The details of the construction of this circuit are given in Appendix~\ref{sec: Hadamard gate derivation}, however we can summarize the steps involved as follows:
\begin{enumerate}
    \item Prepare $a^\prime+1$ ancilla qubits as $\ket{0}^{a^\prime+1}$, where $a^\prime = \mathcal{O}(\mathrm{log}_2(\tfrac{\Delta}{\sigma}))$.
    \item Encode information about the sign of the position of the oscillator $x$ as $f(x) = \tfrac{1}{2}(1 - \mathrm{sgn(x)})$ into the $(a^\prime+1)$th qubit, where sgn$(x)$ is the sign function.
    \item Encode information about $|\tilde{k}^{(j)}|$ into the $a^\prime$ anti-padding qubits.
    \item Disentangle the qumode from the $a^\prime+1$ qubits using a sequence of double-controlled displacement gates.
The final state after applying the anti-padding gate is shown in Appendix~\ref{sec: Hadamard gate derivation} to be of the form:
\begin{equation}
\begin{split}
    \left[\frac{1}{\sqrt{2^{n/m}}}\sum_{y^{(j)}=0}^{2^{n/m}-1}e^{2\pi iy^{(j)}x^{(j)}/2^{n/m}}\ket{y^{(j)},\Delta^\prime}_{O_j}^{gauss}\right]\otimes \ket{\tilde{\psi}_{\text{final}}}_Q
\end{split}
\end{equation}where $\ket{\tilde{\psi}_{final}}$ is an $a^\prime+1$-qubit state that can be discarded or reset. We can recognize the qumode state as the logical action of the qudit Hadamard on a basis state $\ket{x^{(j)},\Delta}_{O_j}^{gauss}$ up to a change in the lattice spacing $\Delta$. The sequence of gates which implements this logical qudit Hadamard on the oscillator is depicted in Fig.~\ref{fig: Hadamard gate circuit main body}.

\begin{figure*}[ht!]
    \centering
    \begin{quantikz}[wire types={q,b}, column sep = 0.3cm]
        \lstick{$\ket{+}^{\otimes a}$}& \gate[2]{U_p^{(j)}(\Delta)} & \gate[]{reset} &\setwiretype{n}& &\lstick{$\ket{0}^{\otimes a^\prime+1}$} \setwiretype{q}&\gate[2]{U_{ap}^{(j)}(\Delta^\prime)} & \gate[]{reset}\\
        \lstick{$\ket{x^{(j)}, \Delta}_{O_j}^{gauss}$}& & \gate[]{D(-\tfrac{2^{n/m+a}}{2}\Delta)} & \gate[]{F} & & & & \rstick{$H^{(2^{n/m})}_j\ket{x^{(j)}, \Delta^\prime}_{O_j}^{gauss}$}\\
        \end{quantikz}
    \caption{The sequence of gate operations used to realize the qudit Hadamard on a qumode.}
    \label{fig: Hadamard gate circuit main body}
\end{figure*}
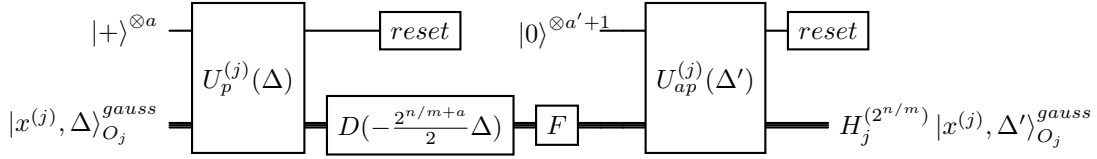

\subsubsection{Controlled rotation gate}

The other qudit gate that we must map onto oscillator gates is the controlled rotation gate $R^{(d)}_k$, whose action on a 2-qudit basis state $\ket{x}\ket{y}$ is given by:
\begin{equation}
    R^{(d)}_k\ket{x}\ket{y} = e^{2\pi i xy/d^k}\ket{x}\ket{y}.
\end{equation}
Note that in the qumode case, qudit basis states are encoded into sharply-peaked Gaussian wavepackets centered about integer multiples of the lattice spacings $\Delta$ and $\Delta^\prime$, where applying a Hadamard gate on one qumode has the effect of transforming this lattice spacing from $\Delta$ to $\Delta^\prime$. Each $j$th layer of the QFT circuit consists of first applying a Hadamard gate to the $j$th oscillator, followed by a sequence of controlled rotation gates on each of the $j^\prime$th qumodes for $j^\prime > j$. Thus, we need a gate on two qumodes which affects the transformation:
\begin{equation}
    R^{d=2^{n/m}}_k\ket{y^{(j)},\tfrac{2\pi}{2^{n/m}\Delta}}^{gauss}_{O_j}\ket{x^{(j^\prime)},\Delta}_{O_{j^\prime}}^{gauss}.
\end{equation}Suppose the width of each of these two wavepackets is small. That is, $\sigma_j,\sigma_{j^\prime}\ll 1$. Then applying the gate $e^{2\pi i\hat{x}_j\hat{x}_{j^\prime}/2^{{nk/m}-1}}$ affects the transformation:
\begin{equation}
\begin{split}
    &e^{2\pi i\hat{x}_j\hat{x}_{j^\prime}/2^{nk/m-1}}\ket{y^{(j)},\tfrac{2\pi}{2^{n/m}\Delta}}^{gauss}_{O_j}\ket{x^{(j^\prime)},\Delta}_{O_{j^\prime}}^{gauss}\\
    &\approx 2^{2\pi iy^{(j)}x^{(j)}/2^{nk/m}}\ket{y^{(j)},\tfrac{2\pi}{2^{n/m}\Delta}}^{gauss}_{O_j}\ket{x^{(j^\prime)},\Delta}_{O_{j^\prime}}^{gauss}.
\end{split}
\end{equation}
\end{enumerate}

\section{Runtime and Bounds on Fidelity}\label{sec: runtime and fidelity bounds}

\subsection{State Transfer Protocols}

In this section we show that a linear increase in the number of qumodes $m$ allows for an exponential improvement in the runtime of state transfer. These results are summarized in Theorem.~\ref{thm: abelian state transfer}.

\begin{theorem}[Parallel state transfer]\label{thm: abelian state transfer}   
    The single-variable abelian D/A state transfer protocol using m qumode and n qubits per qumode achieves a runtime scaling of $\mathcal{O}(2^{n/m}(\Delta + \mathrm{log}(1/\epsilon)))$ with a lower bound of $1-\mathcal{O}(n\epsilon)-me^{-\mathcal{O}(\Delta^2/\sigma^2)}$ on the fidelity, where $\epsilon$ is the error per QSP circuit. The total gate depth is $\mathcal{O}(2^{n/m}\mathrm{log(1/\epsilon}))+\mathcal{O}(n/m)$.

    The D/A state transfer protocol from $n$ qubits to $m$ qumode achieves a runtime scaling of $\mathcal{O}(2^{n/m}\Delta)$ with a fidelity lower bound of:
\begin{equation}
        1 - \sum_{k=1}^m \mathcal{O} \left( \int_{-\infty}^{-\frac{\Delta}{2}(2^{n/m}-1)} dq |\psi_k(q)|^2\\ 
        + \int_{\frac{\Delta}{2}(2^{n/m}-1)}^{\infty} dq |\psi_k(q)|^2 \right) ,
\end{equation}
where $$\psi_k(q) = \frac{1}{\sqrt{\Delta}} \sum_{\mathbf{s_k}} c_{\mathbf{s_k}} \mathrm{sinc} \left( \tfrac{\pi (q-q_{\mathbf{s_k}})}{\Delta} \right)$$
and $\ket{\psi}_Q = \sum_\mathbf{s} c_{\mathbf{s}} |\phi_{\mathbf{s}} \rangle_Q 
= 
\sum_{k=1}^m c_{\mathbf{s_k}} |\phi_{\mathbf{s_k}} \rangle_Q  $. The total gate depth is $\mathcal{O}(2n/m)$.
\end{theorem}

\begin{proof}

Denote the $n$-qubit, $m$-qumode abelian and non-abelian state transfer circuits by $U_{st}(\Delta; n, m)$ and $U_{st}^{na}(\Delta; n, m)$, respectively. $U_{st}(\Delta; n, m)$ and $U_{st}^{na}(\Delta; n, m)$ consist of $m$ copies of the respective single-qumode abelian and non-abelian state transfer circuits running in parallel. We know from Ref.~\cite{liu2024mixed} that for the $n$-qubit, single-qumode case, the runtime for the abelian case scales as $\mathcal{O}(2^n(\Delta + \mathrm{log(1/\epsilon)}))$, the gate depth scales as $\mathcal{O}(2^n\mathrm{log}(1/\epsilon))$, and the infidelity is upper-bounded by $1-\mathcal{O}(n\epsilon)-e^{-\mathcal{O}(\Delta^2/\sigma^2)}$. Here, $\mathcal{O}(n\epsilon)$ is the error contribution from $n$ sequential applications of QSP circuits with error $\epsilon$ each and the $e^{-\mathcal{O}(\Delta^2/\sigma^2)}$ contribution arises from the non-orthogonality of the gaussian wavepackets. The runtime for the non-abelian case has a runtime which scales as $\mathcal{O}(2^n\Delta)$, a gate depth which scales as $\mathcal{O}(n)$, and an infidelity upper bound of:
\begin{equation*}
    1 - \mathcal{O}\left(\int_{-\infty}^{-\tfrac{\Delta}{2}(2^n-1)}dq|\psi(q)|^2 + \int_{\tfrac{\Delta}{2}(2^n-1)}^{\infty}dq|\psi(q)|^2\right).
\end{equation*}Theorem ~\ref{thm: abelian state transfer} is a direct consequence of an application of the triangle inequality to the case where we have $m$ copies of these subroutines running in parallel with $n/m$ qubits per qumode. That is, suppose the errors of the abelian and non-abelian state transfer circuits with $n/m$ qubits and one qumode are given by $\epsilon_{sm,a}$ and $\epsilon_{sm,na}$, respectively. It follows from the triangle inequality that upper bounds on the errors of the $n$-qubit, $m$-qumode (with $n/m$ qubits per qumode) state transfer procedures are given by $\epsilon_{a}\leq m\epsilon_{sm,a}$ and $\epsilon_{na}\leq m\epsilon_{sm,na}$, respectively.
\end{proof}

\subsection{The Quantum Fourier Transform on Qumodes}

In what follows, we formulate a lower bound on the fidelity of the multi-qumode QFT using abelian state transfer by evaluating the error contribution from each step individually, then using the triangle inequality to sum up all these contributions. This is summarized in Theorem~\ref{thm: QFT theorem}.

\begin{theorem}[QFT with parallel abelian state transfer]\label{thm: QFT theorem}
    The QFT circuit on $n$ qubits using m qumodes and $n/m$ qubits per qumode can be implemented with runtime $\mathcal{O}(m2^{n/m+\max(a,a^\prime)}(\Delta+\mathrm{log}(1/\epsilon)))$ + $\mathcal{O}(m^2)$ with a fidelity lower bound of $1-\mathcal{O}(m/2^a)-\mathcal{O}(m\tfrac{\sigma}{\Delta})-\mathcal{O}((n+ma+2ma^\prime)\epsilon)-\mathcal{O}(m\cdot\mathrm{erfc}(\tfrac{\Delta}{\sigma}2^{a^\prime}))$, where $\epsilon$ is the QSP polynomial approximation error.
\end{theorem}

\begin{proof}
The $m$-qumode QFT protocol utilizes $m$ instances of the padding gate carried out sequentially. The padding gate is a modified version of the abelian single-qumode state transfer procedure. Thus, the runtime and fidelity analysis of Ref.~\cite{liu2024mixed} carries over in a directly analogous way. That is, each padding gate sequence consists of $a$ QSP sequences, each with error $\epsilon$ and requires a maximum displacement of $2^{n/m+a-1}\Delta$. Therefore, each padding gate contributes runtime complexity $\mathcal{O}(2^{n/m+a}(\Delta+\mathrm{log}(1/\epsilon)))$ and error $\mathcal{O}(a\epsilon) + e^{-\mathcal{O}(\Delta^2/\sigma^a)}$. There are $m$ padding gates, thus the total runtime contribution from padding is $\mathcal{O}(m2^{n/m+a}(\Delta+\mathrm{log}(1/\epsilon)))$ with error $\mathcal{O}(ma\epsilon)+me^{-\mathcal{O}(\Delta^2/\sigma^2)}$.

We also require $m$ instances of the anti-padding gate carried out sequentially. Each anti-padding gate sequence requires $a^\prime$ anti-padding qubits, a maximum displacement of $2^{n/m+a^\prime-1}\Delta^\prime$ (where $\Delta^\prime = \tfrac{2\pi}{2^{n/m}\Delta}$), $2a^\prime$ controlled QSP sequences of error $\epsilon$ each, and one QSP sequence approximating a step function at the origin of the qumode position. We assume the one step function QSP circuit to have negligible runtime and infidelity compared to the aggregate contribution of the $2a^\prime$ controlled QSP sequences. Each of the $2a^\prime$ controlled QSP sequences implements the same square-wave polynomials as the forward padding gate, but with $a^\prime \neq a$ qubits. Therefore, the runtime complexity of these QSP sequences is $\mathcal{O}(2^{n/m+a^\prime}(\Delta+\mathrm{log}(1/\epsilon)))$ with infidelity $\mathcal{O}(a^\prime\epsilon)+\mathcal{O}(\tfrac{\Delta}{\sigma})+\mathcal{O}(\mathrm{erfc}(\tfrac{\Delta}{\sigma}2^{a^\prime}))+e^{-\mathcal{O}(\Delta^2/\sigma^2)}$. Here, we have assumed $\Delta^\prime < \Delta$ for simplicity. The $\mathcal{O}(\mathrm{erfc}(\tfrac{\Delta}{\sigma}2^{a^\prime}))$ comes from the fact that the gaussian wavepackets in Eq.~\eqref{eq: post Fourier gate H main body} extend out to infinity, yet we only use a finite number of $a^\prime$ anti-padding qubits to remove the periodicity from the qumode. The $\mathcal{O}(\sigma/\Delta)$ contribution is a result from Ref.~\cite{liu2024mixed} that arises from the qubits only being perfectly disentangled from the qumode in the inverse state transfer step in the limit of small $\sigma/\Delta$. The same approximation is used in the current work for the anti-padding step. 

The circuit also involves the application of $\mathcal{O}(m^2)$ controlled rotation gates, each of which is assumed to be implemented with negligible error (assuming $\sigma\ll 1)$ and constant runtime. The runtime contribution from these gates is therefore $\mathcal{O}(m^2)$.

The aggregate of all these runtime and infidelity contributions gives us the scaling in Theorem~\ref{thm: QFT theorem}.

\end{proof}

\section{State transfer can suppress position displacement errors}\label{sec: error suppression}

Here we show that the state transfer procedure, together with its inverse, can correct position displacement errors on an oscillator, provided that the state transfer process itself is noise-free and is implemented with high fidelity. We emphasize that because we currently ignore errors that occur during the state transfer, we have presently not demonstrated this to be a fully fault-tolerant error correction scheme. Nevertheless, one may view this as an interesting example of how using qubits to effectively discretize the continuous set of errors that can occur on a qumode can enable the suppression of such errors. We illustrate this with the simple example of $n$ qubits and one qumode, however we note that this can be readily applied to the multi-qumode state transfer method.

We begin with the state:
\begin{equation}
    \ket{\psi}_O\ket{0}^{\otimes n}_Q = \sum_{x=0}^{2^n-1}C_x\ket{\Delta x}^{gauss}_O\ket{0}^{\otimes n}_Q
\end{equation}where $\ket{\Delta x}^{gauss}_O$ is a Gaussian wavepacket with mean $\Delta x$, with $x$ being an integer and $\Delta$ being any positive real number. Note the slight difference in notation from the previous sections where this same state was denoted as $\ket{x, \Delta}_O^{gauss}$. Suppose a position displacement error of $\delta$ occurs on the qumode, producing the state:
\begin{equation}
\begin{split}
    D(\delta)\ket{\psi}_O\ket{0}^{\otimes n}_Q &= \sum_{x=0}^{2^n-1}C_x\ket{\Delta x + \delta}^{gauss}_O\ket{0}^{\otimes n}_Q\\
    &= \sum_{x=0}^{2^n-1}C_x\ket{\Delta(x+\tfrac{\delta}{\Delta})}^{gauss}_O\ket{0}^{\otimes n}_Q.
\end{split}
\end{equation}The net effect of the displacement error is to shift every integer $x$ in our lattice of Gaussian wavepackets by $\tfrac{\delta}{\Delta}$, keeping the lattice spacing $\Delta$ the same. The first step of the A/D state transfer procedure is to apply a sequence of $n$ hybrid QSP circuits, where the $j$th QSP sequence flips the $j$th qubit conditioned on the position of the oscillator. This is implemented by encoding a square wave pulse:
\begin{equation}
    S_j(\hat{x})= \Theta\left(\mathrm{cos}\left[\tfrac{\pi}{2^{n-j}}(\tfrac{\hat{x}}{\Delta}-2^{n-j-1}+\tfrac{1}{2})\right]\right)
\end{equation}into the diagonal elements of the SU(2) representation of the $j$th QSP sequence:
\begin{equation}
    R_j(\hat{x}) = \begin{pmatrix}
        S_j(\hat{x}) & \sqrt{1-S_j(\hat{x}})\\
        -\sqrt{1-S_j(\hat{x})} & S_j(\hat{x})
    \end{pmatrix}.
\end{equation}Note that among these $n$ square wave pulse sequences, the minimum period is $2\Delta$. Therefore, for all $x=0,1,\dots,2^{n-1}$ and $j=1,2,\dots,n$, we have that $S_j(x \pm \delta^\prime) = S_j(x)$ for $|\delta^\prime| < \tfrac{\Delta}{2}$. Thus, provided that each QSP sequences is implemented with negligible error and assuming that the variance of each Gaussian wavepacket is much smaller than the lattice spacing, \emph{i.e.} $\sigma \ll \tfrac{\Delta}{2}$, the QSP sequences have the following action for $\tfrac{\delta}{\Delta} < \tfrac{1}{2}$:
\begin{equation}\prod_{j=1}^nR_j(\hat{x})\ket{\Delta(x+\tfrac{\delta}{\Delta})}_O^{gauss}\ket{0}^{\otimes n}_Q = \ket{\Delta(x+\tfrac{\delta}{\Delta})}_O^{gauss}\ket{x}_Q.
\end{equation}That is, no information about the displacement is transferred to the qubits during the A/D state transfer.

The next step of the A/D state transfer procedure is to apply a sequence of $n$ conditional displacements to the qumode, where the $j$th conditional displacement displaces the qumode by $-2^{n-j}\Delta$ conditioned on the $j$th bit of $x$ being $1$. This has the net effect:
\begin{equation}
\begin{split}
    \prod_{j=1}^ncD_j(-2^{n-j}\Delta)\ket{\Delta(x+\tfrac{\delta}{\Delta})}_O^{gauss}\ket{x}_Q = \ket{\delta}_O^{gauss}\ket{x}_Q.
\end{split}
\end{equation}Therefore, the action of the A/D state transfer circuit on the displaced state is given by:
\begin{equation}
    U_{A/D}(\Delta)D(\delta)\ket{\psi}_O\ket{0}^{\otimes n}_Q = \ket{\delta}_O^{gauss}\otimes\sum_{x=0}^{2^n-1}C_x\ket{x}_Q.
\end{equation}Provided that we have the ability to reset the qumode to $\ket{0}^{gauss}_O$, the application of this reset followed by the D/A state transfer circuit, we will have the state:
\begin{equation}
    U_{D/A}G_{reset}U_{A/D}(\Delta)D(\delta)\ket{\psi}_O\ket{0}^{\otimes n}_Q = \ket{\psi}_O\ket{0}^{\otimes n}_Q.
\end{equation}Thus, the state transfer procedure and its inverse can suppress some position displacement errors. We give an example of how this can be used to correct displacement errors in cat states in Appendix~\ref{sec: correct cat state}

\section{Exponential Suppression of photon loss rate}\label{sec: photon loss rate}

We can also make statements about how the instantaneous time rate rate of change of photon loss between the abelian single-mode sequential state transfer and multi-mode parallel state transfer protocols compare. We begin with the following theorem, then procede with the proof.

\begin{theorem}[Photon loss rate upper bound]\label{thm: photon loss}
    The ratio of the maximum instantaneous time rate of change of $\braket{\hat{n}}$ of multi-qumode state transfer to that of single-qumode state transfer is upper bounded by $\mathcal{O}(m2^{2n(\tfrac{1}{m}-1)})$, where $n$ and $m$ are the total number of qubits and qumodes, respectively. The corresponding ratio for the average total number of photons lost is $\mathcal{O}(m/(1 - e^{-\gamma 2^{n}(\Delta + \mathrm{log}(1/\epsilon))}))$.
\end{theorem}
\begin{proof}
    We know that at any given point in time, a state with average photon number $\braket{\hat{n}}$ will lose photons at a rate:
\begin{equation}\label{eq: photon loss rate}
    \frac{d}{dt}\braket{\hat{n}} = -\gamma\braket{\hat{n}}
\end{equation}
for some $\gamma > 0$. For simplicity we assume that the amount of squeezing between the single and multi-qumode state transfer procedures is the same and neglect the contribution from squeezing to the average photon number. The average photon number for a state $\ket{x,\Delta}_{O}^{gauss}$ is given by $|x\Delta|^2$. For the abelian single-qumode state transfer, the maximum displacement used is $\mathcal{O}(2^n\Delta)$. For the abelian multi-qumode state transfer, the maximum displacement for any single qumode is $\mathcal{O}(2^{n/m}\Delta)$. There are $m$ such qumodes. An upper bound on the average number of photons for the single-qumode case is therefore $\braket{\hat{n}}_{sm} = \mathcal{O}(2^{2n}\Delta^2).$ The corresponding average photon number for the multi-qumode case is $\braket{\hat{n}}_{mm} = \mathcal{O}(m2^{2n/m}\Delta^2)$. The ratio of these two quantities is given by:
\begin{equation*}
    \frac{\braket{\hat{n}}_{mm}}{\braket{\hat{n}}_{sm}} = \frac{\mathcal{O}(m2^{2n/m}\Delta^2)}{\mathcal{O}(2^{2n}\Delta^2)} = \mathcal{O}(m2^{2n(\tfrac{1}{m}-1)}).
\end{equation*}We can also make a statement about the average total number of photons lost between the sequential and parallel state transfer methods. Eq.~\eqref{eq: photon loss rate} is solved by the expression:
\begin{equation}
    \braket{\hat{n}}_t = \braket{\hat{n}}_{t=0}e^{-\gamma t}
\end{equation}Here, we define time $t=0$ to be the instant of time immediately following the first conditional displacement gate as this corresponds to when the qumodes for each of the respective methods have half of their maximum photon counts. $\braket{\hat{n}}_{t=0,mm} = \mathcal{O}(m2^{2n/m}\Delta^2)$ and $\braket{\hat{n}}_{t=0,sm} = \mathcal{O}(2^{2n}\Delta^2)$ for the single and multi-qumode cases, respectively. We also know from Ref.~\cite{liu2024mixed} and from Theorem~\ref{thm: abelian state transfer} that the runtimes of the sequential and parallel cases are given by $\mathcal{O}(2^n(\Delta + \mathrm{log}(1/\epsilon)))$ and $\mathcal{O}(2^{n/m}(\Delta + \mathrm{log}(1/\epsilon)))$, respectively. Therefore, the average total number of photons lost by the multi-qumode state transfer is $\mathcal{O}(m\Delta^2(1 - e^{-\gamma 2^{n/m}(\Delta + \mathrm{log}(1/\epsilon))}))$ and the total number of photons by the sequential state transfer is $\mathcal{O}(\Delta^2(1 - e^{-\gamma 2^n(\Delta + \mathrm{log}(1/\epsilon))}))$. The ratio of these two quantities is given by:
\begin{equation}
\begin{split}
    \frac{\braket{\hat{n}}_{t,mm}}{\braket{\hat{n}}_{t,sm}} &= \mathcal{O}\left(m\frac{1 - e^{-\gamma 2^{n/m}(\Delta + \mathrm{log}(1/\epsilon))}}{1 - e^{-\gamma 2^n(\Delta + \mathrm{log}(1/\epsilon))}}\right)\\
    &=\mathcal{O}(m/(1 - e^{-\gamma 2^{n}(\Delta + \mathrm{log}(1/\epsilon))})).
\end{split}
\end{equation}
\end{proof}We can see that in the limit of $2^n\gg\gamma$, the total average photon loss suppression is approximately linear in $m$, which we present as a corollary.
\begin{corollary}
    The ratio of the average photon loss of multi-qumode state transfer to that of single-qumode state transfer is approximately $\mathcal{O}(m)$ in the limit of $2^n\gg\gamma$.
\end{corollary}

\section{Numerical demonstration of abelian state transfer}\label{sec: state transfer numerical simulation}

In this section we present a simple numerical demonstration of the abelian state transfer procedure involving two qubits and two qumodes using QuTiP.~\cite{qutip5} Our initial 2-qubit state is a bell state and our initial 2-qumode state is a tensor product of two squeezed vacuum states. This state is given by:
\begin{equation}
    \ket{\psi_{init}} = \frac{1}{\sqrt{2}}\left[\ket{00}_Q+\ket{11}_Q\right]\otimes\ket{0,\Delta}^{gauss}_{O_1}\ket{0,\Delta}_{O_2}^{gauss}.
\end{equation}We choose $\Delta=1$, a squeezing parameter of $1.2$, and a Fock level cutoff of $80$. The target state we wish to prepare is given by:
\begin{equation*}
\begin{split}
    &\ket{\psi_{targ}} =\ket{00}_Q\\
    &\otimes\frac{1}{\sqrt{2}\mathcal{N}}\left[\ket{0,\Delta}_{O_1}^{gauss}\ket{0,\Delta}^{gauss}_{O_2}
    +\ket{1,\Delta}_{O_1}^{gauss}\ket{1,\Delta}_{O_2}^{gauss}\right].
\end{split}
\end{equation*}We measure the fidelity of the protocol as the fidelity between the reduced density matrix of the qumode subsystem with the state:
\begin{equation}
\begin{split}
    \ket{\psi_{targ,osc}} = \frac{1}{\sqrt{2}\mathcal{N}}&[\ket{0,\Delta}_{O_1}^{gauss}\ket{0,\Delta}^{gauss}_{O_2}\\
    +&\ket{1,\Delta}_{O_1}^{gauss}\ket{1,\Delta}_{O_2}^{gauss}].
\end{split}
\label{eq:target-osc}
\end{equation}We report that the fidelity of the method using the aforementioned parameters is $\mathcal{F}\approx0.9993$, validating the method. Additionally, we provide a visual demonstration of the method by plotting the Wigner functions of the following two single-qumode density matrices:
\begin{align}
    &\rho_{2}^{(1)} = \frac{1}{\mathcal{N}_1}\mathrm{Tr}\left(\ket{0,\Delta}\bra{0,\Delta}_{O_1}^{gauss}\otimes \mathbb{I}_{O_2}\ket{\psi_{osc}}\bra{\psi_{osc}}\right)\\
    &\rho_{2}^{(2)} = \frac{1}{\mathcal{N}_2}\mathrm{Tr}\left(\ket{1,\Delta}\bra{1,\Delta}_{O_1}^{gauss}\otimes \mathbb{I}_{O_2}\ket{\psi_{osc}}\bra{\psi_{osc}}\right)
\end{align}where $\ket{\psi_{osc}}$ is the state that we prepare which is intended to have high fidelity with $\ket{\psi_{targ,osc}}$. We expect to see wavepackets centered at $x = 0$ and $x=\Delta$ for $\rho_{2}^{(1)}$ and $\rho_{2}^{(1)}$, respectively. 
\begin{figure*}[htp!]
    \centering
    \begin{subfigure}[b]{0.45\textwidth}
        \includegraphics[width=\linewidth]{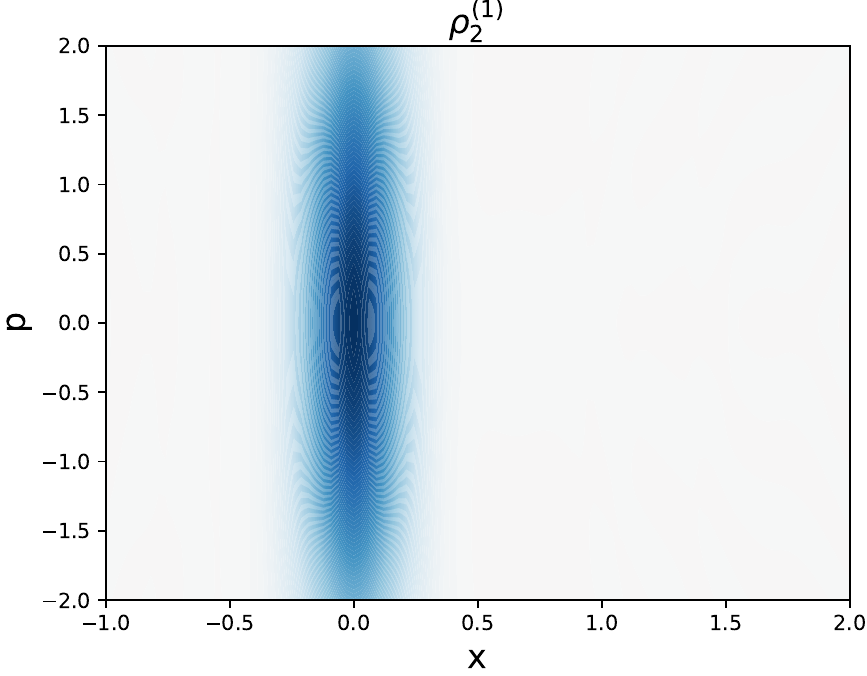}
    \end{subfigure}
    \begin{subfigure}[b]{0.45\textwidth}
        \includegraphics[width=\linewidth]{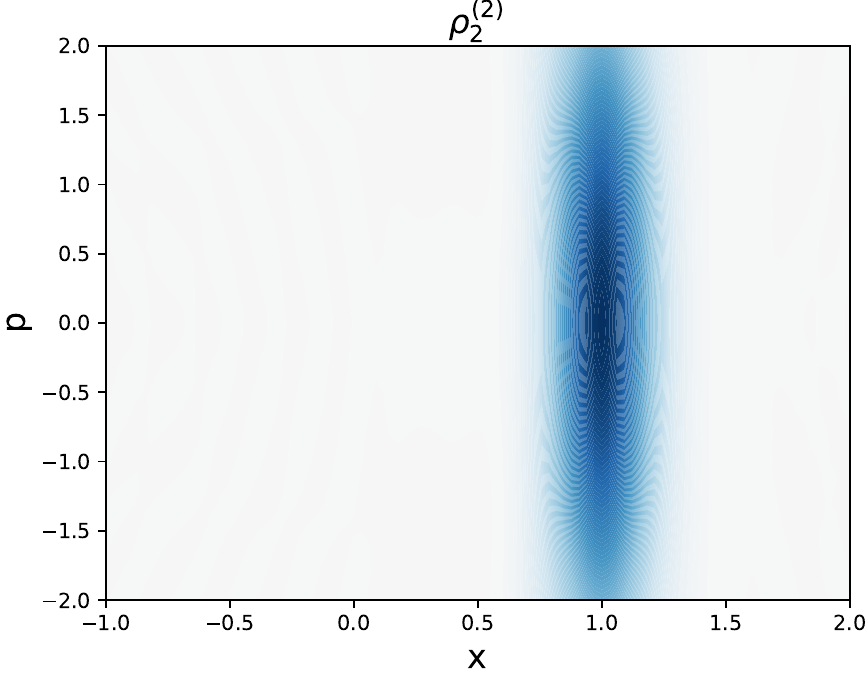}
    \end{subfigure}
    \caption{The Wigner functions of the reduced density matrix of the second qumode after projecting the first qumode onto the state $\ket{0,\Delta}_{O_1}^{gauss}$ (left) and $\ket{1,\Delta}_{O_1}^{gauss}$ (right), which agrees with Eqs. \eqref{eq:target-osc}.}
    \label{fig:placeholder}
\end{figure*}

\section{Comparison to qubit QFT}\label{sec: comparison to qubit qft}

One question is: are there parameter regimes in which it is advantageous to use this method over the qubit QFT method which scales as $\mathcal{O}(n^2)$ with the number of qubits $n$? Because our derivation of QFT includes only loose lower bounds on the fidelity, it is difficult to make unambigous conclusions in this regard. However, by using some simplifying assumptions, we can still make some limited statements with a ``back-of-the-envelope" estimation. The runtime of qumode QFT is (ignoring constant factors):
\begin{equation*}
\begin{split}
    T_{osc} &\approx m2^{n/m+a}(\Delta+\mathrm{log}(1/\epsilon_{QSP}))\\
    &+ m2^{n/m+a^\prime}(\Delta+\mathrm{log}(1/\epsilon_{QSP})) + m^2.
\end{split}
\end{equation*}Setting the qudit level $d = 2^{n/m}$ and insisting that the qubit and qudit Hilbert spaces are the same size gives us:
\begin{equation*}
    m = \frac{n}{\mathrm{log}_2(d)}.
\end{equation*}Suppose we want the total error $\epsilon$ to be on the order of $m\tfrac{\sigma}{\Delta}\approx\tfrac{m}{2^a}\approx\tfrac{m}{2^{a^\prime}}$. The $\mathrm{erfc}(\cdot)$ contribution to the error implies that we require $2^{a^\prime}\approx\tfrac{m}{\epsilon}$. This gives us a runtime of approximately:
\begin{equation*}
\begin{split}
    T_{osc} &\approx 2m^2d(\Delta+\mathrm{log}(1/\epsilon_{QSP})) + m^2\\
    &\approx 2m^2d(\Delta+\mathrm{log}(1/\epsilon_{QSP})).
\end{split}
\end{equation*}Denote the average primitive gate times for the qubit and hybrid qubit-oscillator systems by $T_{gate,q}$ and $T_{gate,o}$, respectively. Estimating the qubit QFT runtime as $T_{qubit}\approx n^2=m^2\mathrm{log}^2_2(d)$ gives us a qumode speedup of:
\begin{equation*}
    \frac{T_{qubit}}{T_{osc}} \approx \frac{T_{gate,q}}{T_{gate,o}}\left(\frac{\epsilon\mathrm{log}_2^2(d)}{d(\Delta+\mathrm{log}(1/\epsilon_{QSP}))}\right).
\end{equation*}Thus, the qumode QFT has a lower runtime than qubit QFT when the following condition is satisfied:
\begin{equation}\label{eq: speedup condition}
    \frac{T_{gate,q}}{T_{gate,o}} \geq 2\frac{d}{\epsilon}\left(\frac{\Delta+\mathrm{log}(1/\epsilon_{QSP})}{\mathrm{log}_2^2(d)}\right)
\end{equation}Thus, we remark that qumode QFT offers a speedup over (exact) qubit QFT in some combination of the following regimes:
\begin{enumerate}
    \item The runtime of qubit primitive gates is much higher than that of primitive qubit-oscillator hybrid gates (on average).
    \item The qudit level $d$ is much smaller than the inverse error $\tfrac{1}{\epsilon}$.
\end{enumerate}

As a quick illustrative example, we can consider the hypothetical case where $\epsilon_{QSP}\approx 10^{-3}$, in which case $\Delta + \mathrm{log}_2(1/\epsilon_{QSP})\approx11$ when $\Delta = 1$. Plotting the RHS of Eq.~\eqref{eq: speedup condition} as a function of $(d,\epsilon)$ gives us a contour map where the level curves tell us what ratio $\tfrac{T_{gate,q}}{T_{gate,o}}$ would be necessary to reach a break-even point where qumode QFT is faster than qubit QFT. This resource comparison is given in Fig~\ref{fig: resource comparison figure}.
\begin{figure}[htbp!]
    \centering
    \includegraphics[width=\linewidth]{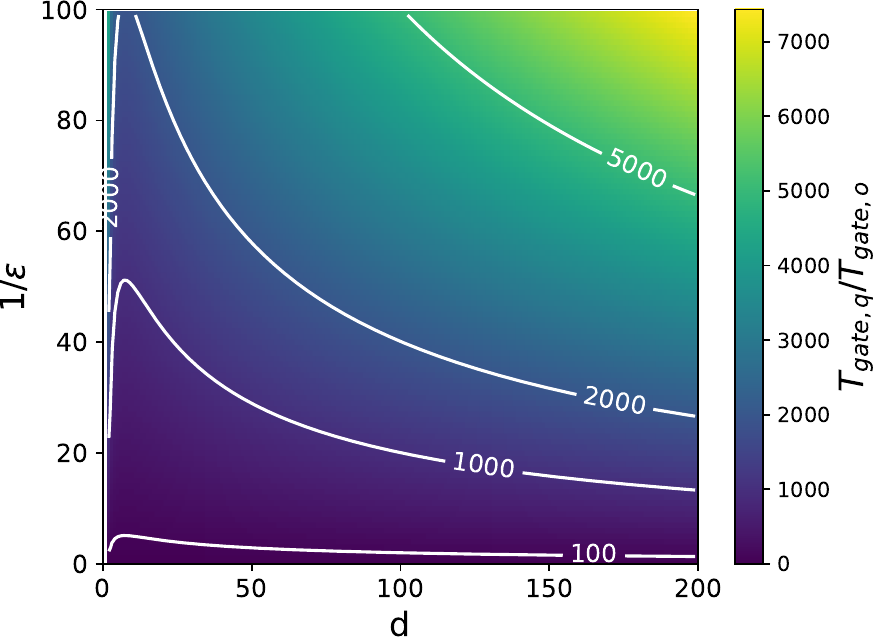}
    \caption{The gate time ratio $\tfrac{T_{gate,q}}{T_{gate,o}}$ required to achieve a break-even speedup of qumode QFT over qubit QFT is a function of the qudit dimension $d$ and the loose upper bound on the error $\epsilon$. When these average primitive gate times are equal, the level curves can be interpretted as the inverse of the qumode QFT speed-up.}
    \label{fig: resource comparison figure}
\end{figure}We can see that in order for qumode QFT to have a speedup over exact qubit QFT for a lower bound on the error of $\approx 10^{-2}$ or smaller, the qudit dimension must be small and the qumode gate times must be at least 3 orders of magnitude faster than the analogous qubit gate times. However, we emphasize that this resource comparison utilizes a loose upper bound on the error $\epsilon$, which likely leads to an overestimation of the runtime of qumode QFT, given that its runtime scales as $\mathcal{O}(\tfrac{1}{\epsilon})$.

\section{Conclusion}\label{sec: conclusion}

In this work, we demonstrated a method of transferring a multi-qubit state to a multi-qumode state. This method has a provably exponential speed-up over single-qumode state transfer~\cite{liu2024mixed} at the cost of a linear increase in the number of qumodes needed. That is, the $\mathcal{O}(2^n)$ term in the single-qumode case becomes $\mathcal{O}(2^{n/m})$ for the $m$-qumode case. The method is flexible in the sense that the number of qumodes we allocate to this procedure is a user-input variable that comes with a quantifiable trade-off.

We have also shown that this parallel state transfer has an improved robustness against photon loss error as compared to the single-mode protocol. Furthermore, both the single and multi-qumode abelian state transfer procedures possess modest intrinsic resistance to position displacement errors, which can be intuited as arising from the "GKP-like" structure of the QSP polynomials used. 

We have demonstrated how this state transfer procedure can be used as a subroutine in hybrid CV-DV quantum computation, using the quantum Fourier transform as an illustrative example, which likewise achieves an exponential speedup over its single-qumode counterpart. We derived an approximate condition which must be met in order for multi-qumode QFT to be advantageous in terms of runtime over exact qubit QFT. This resource estimate shows that qumode QFT typically has a runtime advantage in the regime of small numbers of qubits per qumode, low algorithmic fidelity, and oscillator gates that are at least two orders of magnitude faster than their qubit counterparts. We note that a large contribution to the qumode QFT runtime as a function of the algorithmic fidelity is primarily attributable to the necessity of padding qubits. That is, an error scaling of $\mathcal{O}(1/2^a)$ and a runtime scaling of $\mathcal{O}(2^a)$ imply a runtime scaling of $\mathcal{O}(1/{\epsilon_{alg}})$ in the algorithmic error $\epsilon_{alg}$.

The state transfer procedure in this work represents a step forward towards scalable state transfer, which will be of particular interest for scalable fault-tolerant distributed quantum computation, where the high-fidelity transfer of exponentially large Hilbert spaces with low latency will be critical. We note the potential advantages the multi-mode state transfer has over another promising method: using entangled bits (ebits) as a resource to perform teleportation and nonlocal gates.~\cite{Eisert2000} Such methods typically require a process known as entanglement distillation~\cite{Bennett_ebit_distillation, qLDPC_bell_state_distillation}, which consumes many low-fidelity Bell pairs to produce a single high-fidelity Bell pair. This can incur a large qubit overhead.

The state transfer protocol in the current work is resource-efficient in the sense that the number of qumodes required is upper-bounded by the number of qubits whose state is being transferred and no distillation procedure is necessary. However, the current work does not present a thorough resource estimate comparison to entanglement distillation or ion shuttling. Such a resource estimate on the scale of fault-tolerant quantum computation would be an interesting future direction of research that would provide value in guiding distributed quantum computing architecture design.

It remains unclear to what extent the ability of state transfer to suppress some position displacement errors will be useful in practice when considering realistic scenarios where errors occur during the state transfer procedure itself. Delving deeper into the noise-resilience of this procedure and how it can be combined with error correction codes that correct momentum displacement errors (such as a biased GKP code~\cite{stafford_biased_2023}) would be an interesting future direction of research.

Finally, we note that while the runtime of the multi-qumode state transfer is scalable, its unclear to what extent the same is true for computational primitives utilizing state transfer such as QFT. In order for the qumode QFT to be competitive with its qubit state-of-the-art counterparts, it will be crucial to develop alternative means of encoding a qudit into the phase space of a qumode in such a way that the use of padding qubits is unnecessary.

Suppose the above goal of eliminating padding qubits be achieved. The next logical step beyond this would be to incorporate state transfer into compilation schemes wherein algorithm subroutines that are expensive on qubits are offloaded to qumodes, then transferred back to qubits for further processing. We look forward to seeing what implications this state transfer procedure has for the future of hybrid qubit-qumode computation.

\section*{Acknowledgements}

The authors thank Frank Mueller, John Stack, and Cheng Chu for helpful discussions and John Martyn for helpful comments on the manuscript. This work is supported by the U.S. Department of Energy, Office of Science, Advanced Scientific Computing Research, under contract number DE-SC0025384 as well as NSF grants OMA-2120757, PHY-2325080, and OSI-2410675.

\appendix

\section{Multi-Mode QFT circuit}\label{sec: appendix QFT derivation}

\subsection{Full derivation of Hadamard gate}\label{sec: Hadamard gate derivation}

We begin by noting the similarity of Eq~\eqref{eq: qudit Hadamard gate} with the qubit QFT operation that was implemented using a free evolution gate for the single-qumode case. Recall from the single-qumode QFT that the free evolution gate only produced the needed action when the qumode state was periodic. A similar statement is true for our current case. Our $m$ qumodes do not obey modulo-$d$ arithmetic, therefore we must make our qumode wavefunction itself periodic when applying the free evolution gate, followed by an operation which then removes this periodicity.

Suppose we have a single Gaussian wavepacket $\ket{x^{(j)}}_{O_j}^{gauss}$ together with a collection of qubits in the state $\ket{+}^{\otimes a}$. We require a gate sequence $U_p^{(j)}(\Delta)$ which has the action:
\begin{equation}
\begin{split}
    &U_{p}^{(j)}\ket{+}^{\otimes a}\otimes\ket{x^{(j)},\Delta}_{O_j}^{gauss}\\
    &= \ket{0}^{\otimes a}\otimes\frac{1}{\sqrt{2^a}}\sum_{k^{(j)}=0}^{2^a-1}\ket{2^{n/m}k^{(j)} + x^{(j)},\Delta}^{gauss}_{O_j}.
\end{split}
\end{equation}This is equivalent to requiring that:
\begin{equation}
    U_{p}^{(j)}\ket{k^{(j)}}\ket{x^{(j)},\Delta}_{O_j}^{gauss} = \ket{0}^{\otimes a}\ket{2^{n/m}k^{(j)}+x^{(j)},\Delta}_{O_j}^{gauss}.
\end{equation}
This gate, which we will refer to as a "padding gate", bears a resemblance to the single-qumode state transfer circuit in the sense that we transfer information $k^{(j)}$ from the qubits to the qumode. The difference with the single-qumode case is that the qumode already contains non-zero information. Our circuit construction must take this into account.

Recall that the first stage of the state transfer procedure was to entangle the qumode with the qubits using a sequence of controlled displacement operations. What we want here is an entangling gate that affects the transformation:
\begin{equation}
\begin{split}
    &U_{p,e}^{(j)}(\Delta)\ket{k^{(j)}}\ket{x^{(j)},\Delta}_{O_j}^{gauss}\\
    &= \ket{k^{(j)}}\ket{2^{n/m}k^{(j)} + x^{(j)},\Delta}_{O_j}^{gauss}.
\end{split}
\end{equation}We can implement this using a similar sequence of controlled displacements as was used for the single-qumode case. We simply note that we want to displace the qumode by a total amount $2^{n/m}k^{(j)}\Delta$. $k^{(j)}$ is an integer whose binary expansion is $k^{(j)}=\sum_{l_j=1}^ak^{(j)}_{l_j}\cdot 2^{a-l_j}$. Thus, we can affect the desired transformation using the sequence of controlled displacement gates:
`\begin{align}
    &\bigotimes_{l_j=1}^{a}cD_{l_j}(2^{{n/m}+a- l_j}\Delta)\ket{k^{(j)}}\ket{x^{(j)},\Delta}_{O_j}\\
    &= \ket{k^{(j)}}\ket{2^{n/m}k^{(j)}+x^{(j)},\Delta}_{O_j}^{gauss}.
\end{align}Next we must disentangle the qubits from the qumode. Recall from the single-qumode state transfer that one can use a sequence of QSP circuits, each of which implemented a square wave polynomial that would flip the $l_jth$ qubit conditioned on the $l_j$th bit of the integer $x^{(j)}$ encoded into the position of the oscillator being 1. We need to do something similar here, but only apply such QSP sequences for a subset of the bits of the integer $2^nk^{(j)} + x^{(j)}$. That is, we want to flip the $l_j$th qubit conditioned only on information about $k^{(j)}$ and not $x^{(j)}$. We note that $2^{n/m}k^{(j)} + x^{(j)}$ is an $a+1$-bit integer where bits $a+1,\dots a+n/m$ contain information about only $x^{(j)}$ and bits $1,\dots,a$ contain information about only $k^{(j)}$. Thus, we can apply a sequence of QSP circuits, each of which flips the $l_j$th qubit conditioned on the $l_j$th bit of $2^{n/m}k^{(j)} + x^{(j)}$ being 1. Denoting the $l_j$th such QSP sequence by $U_{p,d}^{(l_j)}(\Delta)$, we find that the entangling and disentangling padding gate steps together affect the transformation:
\begin{equation}
\begin{split}
    &\left[\bigotimes_{l_j=1}^aU_{p,d}^{(p)}(\Delta)\right]\left[\bigotimes_{l_j=1}^acD_{l_j}(2^{n/m+a-l_j}\Delta)\right]\ket{+}^{\otimes a}\ket{x^{(j)},\Delta}_{O_j}^{gauss}\\
    &= \ket{0}^{\otimes a}\frac{1}{\sqrt{2^a}}\sum_{k^{(j)}=0}^{2^a-1}\ket{2^{n/m}k^{(j)} + x^{(j)},\Delta}_{O_j}^{gauss}.
\end{split}
\end{equation}Now that our qumode wavefunction is periodic, we symmetrize it by applying a displacement $D_j(-2^{{n/m}+a-1}\Delta)$:
\begin{equation}
\begin{split}
    &D_j(-2^{{n/m}+a-1}\Delta)\ket{0}^{\otimes a}\frac{1}{\sqrt{2^a}}\sum_{k^{(j)}=0}^{2^a-1}\ket{2^{n/m}k^{(j)} + x^{(j)},\Delta}_{O_j}^{gauss}\\
    &= \ket{0}^{\otimes a}\frac{1}{\sqrt{2^a}}\sum_{k^{(j)}=-2^{a}/2}^{2^a/2-1}\ket{2^{n/m}k^{(j)} + x^{(j)},\Delta}_{O_j}^{gauss}.
\end{split}
\end{equation}Next we apply a free evolution gate which affects the Fourier transform on the qumode state:
\begin{equation}\label{eq: post Fourier gate H}
\begin{split}
    &F_j\frac{1}{\sqrt{2^a}}\sum_{k^{(j)}=-2^{a}/2}^{2^a/2-1}\ket{2^{n/m}k^{(j)} + x^{(j)},\Delta}_{O_j}^{gauss}\\
    &\approx \left(\frac{2\pi\sigma}{\Delta 2^{n/m}}\sqrt{\frac{2}{\pi}}\right)^{1/2}\sum_{\ell^{(j)}=-\infty}^{\infty}e^{2\pi i \ell^{(j)}x^{(j)}/2^{n/m}}\\
    &\times e^{-\left(\tfrac{2\pi}{2^{n/m}\Delta}l^{(j)}\right)^2\sigma^2}\ket{\ell^{(j)}, \Delta^\prime}^{gauss}_{O_j}
\end{split}
\end{equation}where this approximation incurs an error of $\mathcal{O}(1/2^a)$ and $\ket{\ell^{(j)}}$ is defined as:\cite{liu2024mixed}
\begin{equation}
\begin{split}
    \ket{\ell^{(j)}, \Delta^\prime}_{O_j}^{gauss} &= \sqrt{\frac{\Delta 2^{n/m}}{2\pi 2^a}}\int_{\Delta^\prime(\ell^{(j)}-\tfrac{1}{2})}^{\Delta^\prime(\ell^{(j)}+\tfrac{1}{2})}dq^{(j)}\\
    &\times\left[\sum_{k^{(j)}=-2^a/2}^{2^a/2-1}e^{iq^{(j)}2^{n/m}k^{(j)}\Delta}\right]\ket{q^{(j)}}_{O_j}.
\end{split}
\end{equation}
This is nearly the expression that we want, except for the presence of the Gaussian envelopes and the fact that one of the summations runs over $l^{(j)} \in (-\infty, \infty)$ instead of $y^{(j)} = 0,\dots,2^{n/m}-1$. We emphasize that the Gaussian envelopes do not decay quickly enough to automatically account for this discrepancy in summands, especially for small $\tfrac{\sigma}{\Delta}$. For instance, note that for $\tfrac{\sigma}{\Delta} = 0.1$ and $\ell^{(j)}=2 \cdot 2^{n/m}$, $e^{-\left(\tfrac{2\pi}{2^n\Delta}l^{(j)}\right)^2}\approx0.206$. For $\tfrac{\sigma}{\Delta}=10^{-2}$, the Gaussian envelope quantity exceeds $0.9$ for $\ell^{(j)}$ as large as $5\cdot2^{n/m}$. That is, these terms have roughly the same weight as for $\ell^{(j)}\leq2^{n/m}-1$.

Denote $\mathcal{D} = \{y \in \mathbb{Z}|0\leq y \leq 2^{n/m}-1\}$. Suppose we can express $\ell^{(j)}$ as the output of a bijective function $f:\mathbb{Z}\times\mathcal{D}\mapsto \mathbb{Z}$. That is, $f$ takes one integer and one integer in $\mathcal{D}$ and outputs a unique integer in $\mathcal{D}$. We write $\ell^{(j)}(\tilde{k}^{(j)},y^{(j)}) = \ell^{(j)}$ as a shorthand, where $\tilde{k}^{(j)}\in\mathbb{Z}$ and $y^{(j)}\in\mathcal{D}$. We can then re-express Eq.~\ref{eq: post Fourier gate H} as a double summation over $\tilde{k}^{(j)},y^{(j)}$ with any appearance of $\ell^{(j)}$ substituted with its expansion in terms of $\tilde{k}^{(j)},y^{(j)}$. That is, because this mapping is bijective, we will have not added, ommited, or altered the value of any terms in Eq.~\ref{eq: post Fourier gate H}. The function $\ell^{(j)}=2^{n/m}\tilde{k}^{(j)}+y^{(j)}$ is one such bijection. We have used the notation $\tilde{k}^{(j)}$ to avoid any confusion with the original padding index $k^{(j)}$. These two indices play similar roles, but cannot necessarily be interpreted to represent the same quantity.

Performing this substitution, we arrive at the expression:

\begin{widetext}
\begin{equation}\label{eq: post Fourier gate H re-expressed}
\begin{split}
    &\left(\frac{2\pi\sigma}{\Delta 2^{n/m}}\sqrt{\frac{2}{\pi}}\right)^{1/2}\sum_{\tilde{k}^{(j)}=-\infty}^{\infty}\sum_{y^{(j)}=0}^{2^{n/m}-1}e^{2\pi i y^{(j)}x^{(j)}/2^{n/m}}e^{-\left(\tfrac{2\pi}{2^{n/m}\Delta}(2^{n/m}\tilde{k}^{(j)}+y^{(j)})\right)^2\sigma^2}\ket{2^{n/m}\tilde{k}^{(j)}+y^{(j)}, \Delta^\prime}^{gauss}_{O_j}\\
    &\approx \left(\frac{2\pi\sigma}{\Delta 2^{n/m}}\sqrt{\frac{2}{\pi}}\right)^{1/2}\sum_{\tilde{k}^{(j)}=-\infty}^{\infty}\sum_{y^{(j)}=0}^{2^{n/m}-1}e^{2\pi i y^{(j)}x^{(j)}/2^{n/m}}e^{-\left(\tfrac{2\pi}{\Delta}\tilde{k}^{(j)}\right)^2\sigma^2}\ket{(2^{n/m}\tilde{k}^{(j)}+y^{(j)}), \Delta^\prime}_{O_j}^{gauss}
\end{split}
\end{equation}
\end{widetext}where the approximation introduced is to error $\mathcal{O}(\tfrac{\sigma}{\Delta})$~\cite{liu2024mixed}.

Our goal is to transfer information about $\tilde{k}^{(j)}$ back to a multi-qubit register; therefore, we must introduce an approximation which truncates the infinite summation over $\tilde{y}(j)$ to run over integers $-\tfrac{2^{a^\prime}}{2},\dots,\tfrac{2^{a^\prime}}{2}-1$ for some integer $a^\prime$. We denote this integer by $a^\prime$ to emphasize that it may be chosen differently from the original padding integer $a$. We show in Appendix~\ref{sec: Hadamard gate error bounds} that one can truncate this expression in this fashion to error $\mathcal{O}(\mathrm{erfc}(\tfrac{\Delta}{\sigma}2^{a^\prime}))$ to get the approximate expression:
\begin{widetext}
\begin{equation}
    \left(\frac{2\pi\sigma}{\Delta 2^{n/m}}\sqrt{\frac{2}{\pi}}\right)^{1/2}\sum_{\tilde{k}^{(j)}=-2^{a^\prime}/2}^{2^{a^\prime}/2-1}\sum_{y^{(j)}=0}^{2^{n/m}-1}e^{2\pi i y^{(j)}x^{(j)}/2^{n/m}}e^{-\left(\tfrac{2\pi}{\Delta}\tilde{k}^{(j)}\right)^2\sigma^2}\ket{(2^{n/m}\tilde{k}^{(j)}+y^{(j)}), \Delta^\prime}_{O_j}^{gauss}
\end{equation}
\end{widetext}
We can break down the problem of disentangling the padding qubits from the qumode in a two-step black box fashion. We will first verify that these black boxes give us the desired result, then we will explicitly construct gate sequences which realize the action of these black boxes. We call the first black box $U_{ap,ent,j}$ ("anti-padding entangling gate sequence") which first encodes information about the sign of $\tilde{k}^{(j)}$ into a $1$-qubit register and information about the magnitude of $\tilde{k}^{(j)}$ into an $a^\prime$-qubit register. We call the second black box $U_{ap,d,j}$ ("anti-padding disentangling gate sequence") which displaces the qumode by an amount proportional to $\tilde{k}^{(j)}$ controlled on the information about $\tilde{k}^{(j)}$ stored in the $a^\prime+1$-qubit register.

Suppose we have a black box unitary $U_{ap,ent,j}$ which has the action:
\begin{equation}
\begin{split}
    &U_{ap,ent,j}\ket{2^{n/m}\tilde{k}^{(j)}+y^{(j)}}_{O_j}\ket{0}^{\otimes a^\prime+1}_{Q}\\
    &= \ket{2^{n/m}\tilde{k}^{(j)}+y^{(j)},\Delta^\prime}^{gauss}_{O_j}\ket{|\tilde{k}^{(j)}|}_Q\ket{f(\tilde{k}^{(j)})}_Q
\end{split}
\end{equation}where:
\begin{equation}
    f(\tilde{k}^{(j)}) = 
    \begin{cases}
        0 & \text{if } \tilde{k}^{(j)} \geq 0 \\
        1 & \text{if } \tilde{k}^{(j)} < 0.
    \end{cases}
\end{equation}Then:
\begin{widetext}
\begin{equation}\label{eq: black box 1 action}
\begin{split}
    &U_{ap,ent,j}\left(\Delta^\prime\sigma\sqrt{\frac{2}{\pi}}\right)^{1/2}\sum_{\tilde{k}^{(j)}=-2^{a^\prime}/2}^{2^{a^\prime}/2-1}\sum_{y^{(j)}=0}^{2^{n/m-1}}e^{2\pi i y^{(j)}x^{(j)}/2^{n/m}}e^{-(\Delta^\prime \tilde{k}^{(j)})^2\sigma^2}\ket{(2^{n/m}\tilde{k}^{(j)}+y^{(j)}),\Delta^\prime}^{gauss}_{O_j}\ket{0}^{\otimes a+1}\\
    &=\left(\Delta^\prime\sigma\sqrt{\frac{2}{\pi}}\right)^{1/2}\sum_{\tilde{k}^{(j)}=-2^{a^\prime}/2}^{2^{a^\prime}/2-1}\sum_{y^{(j)}=0}^{2^{n/m-1}}e^{2\pi i y^{(j)}x^{(j)}/2^{n/m}}e^{-(\Delta^\prime \tilde{k}^{(j)})^2\sigma^2}\ket{(2^{n/m}\tilde{k}^{(j)}+y^{(j)}),\Delta^\prime}^{gauss}_{O_j}\ket{|\tilde{k}^{(j)}|}_{Q}\ket{f(\tilde{k}^{(j)})}_Q.
\end{split}
\end{equation}
\end{widetext}Now suppose we have a second black box unitary $U_{ap,d,j}$ which has the following action:
\begin{equation}
\begin{split}
    &U_{ap,d,j}\ket{(2^{n/m}\tilde{k}^{(j)}+y^{(j)}),\Delta^\prime}^{gauss}_{O_j}\ket{|\tilde{k}^{(j)}|}_{Q}\ket{f(\tilde{k}^{(j)})}_Q\\
    &= \ket{y^{(j)},\Delta^\prime}_{O_j}^{gauss}\ket{|\tilde{k}^{(j)}|}_{Q}\ket{f(\tilde{k}^{(j)})}_Q.
\end{split}
\end{equation}Then:

\begin{widetext}
\begin{equation}
\begin{split}
    &U_{ap,d,j}\left(\Delta^\prime\sigma\sqrt{\frac{2}{\pi}}\right)^{1/2}\sum_{\tilde{k}^{(j)}=-2^{a^\prime}/2-1}^{2^{a^\prime}/2-1}\sum_{y^{(j)}=0}^{2^{n/m-1}}e^{2\pi i y^{(j)}x^{(j)}/2^{n/m}}e^{-(\Delta^\prime \tilde{k}^{(j)})^2\sigma^2}\ket{(2^{n/m}\tilde{k}^{(j)}+y^{(j)}),\Delta^\prime}^{gauss}_{O_j}\ket{|\tilde{k}^{(j)}|}_{Q}\ket{f(\tilde{k}^{(j)})}_Q\\
    &=\left(\Delta^\prime\sigma\sqrt{\frac{2}{\pi}}\right)^{1/2}\sum_{\tilde{k}^{(j)}=-2^{a^\prime}/2}^{2^{a^\prime}/2-1}\sum_{y^{(j)}=0}^{2^{n/m-1}}e^{2\pi i y^{(j)}x^{(j)}/2^{n/m}}e^{-(\Delta^\prime \tilde{k}^{(j)})^2\sigma^2}\ket{y^{(j)},\Delta^\prime}^{gauss}_{O_j}\ket{|\tilde{k}^{(j)}|}_{Q}\ket{f(\tilde{k}^{(j)})}_Q\\
    &=\left[\frac{1}{\sqrt{2^{n/m}}}\sum_{y^{(j)}=0}^{2^{n/m}-1}e^{2\pi iy^{(j)}x^{(j)}/2^{n/m}}\ket{y^{(j)},\Delta^\prime}_{O_j}^{gauss}\right]\otimes \left[\left(2^{n/m}\Delta^\prime\sigma\sqrt{\frac{2}{\pi}}\right)^{1/2}\sum_{\tilde{k}^{(j)}=-2^{a-1}}^{2^{a-1}-1}e^{-(\Delta^\prime \tilde{k}^{(j)})^2\sigma^2}\ket{|\tilde{k}^{(j)}|}_Q\ket{f(\tilde{k}^{(j)})}_Q\right]
\end{split}
\end{equation}
\end{widetext}where we can see that the qumode state now encodes $H^{(d=2^{n/m})}\ket{x^{(j)},\Delta}_{O_j}^{gauss}$ up to a change in the lattice spacing $\Delta \rightarrow\Delta^\prime$. We can accept this lattice spacing change as we will simply take it into account in the construction of the controlled rotation gates. The question is: How do we construct these two black boxes?

\subsubsection{Constructing the first black box: $U_{ap,ent,j}$}

The construction of this black box is done in two parts: one which maps the $a^\prime+1$th qubit register to $\ket{\tilde{k}^{j)}}$ and one which maps the other $a^\prime$ qubits to $\ket{|\tilde{k}^{(j)}|}$. We begin with the first stage by noting that because $y^{(j)}\leq2^{n/m}-1$, querying whether or not $\tilde{y}^{(j)}$, querying whether or not $\tilde{k}^{(j)}<0$ is equivalent as querying whether or not $2^{n/m}\tilde{k}^{(j)}+y^{(j)}<0$. The complete gate sequence is given by:
\begin{equation}
    D(\tfrac{\Delta^\prime}{2})R_{QSP}D(-\tfrac{\Delta}{2})
\end{equation}where $R_{QSP}$ is the QSP circuit.
We can consider a method by which we construct the entangling stage (the disentangling stage remains the same) of this first black box using quantum signal processing. What we want is a degree $d$ polynomial $P(e^{i\kappa \hat{x}})$ which implements the following block encoding:
\begin{equation}
    R_{\Theta} = \begin{bmatrix}
        \Theta(\hat{x}) & -\sqrt{1 - \Theta(\hat{x})^2}\\
        \sqrt{1 - \Theta(\hat{x})^2} & \Theta(\hat{x})
    \end{bmatrix}.
\end{equation}That is, when $x \geq 0$, the QSP circuit acts as the identity on the qubit and does not flip it. However, when $x < 0$, the qubit gets flipped from $\ket{0}$ to $\ket{1}$. The total circuit is $D_j(-\tfrac{\Delta}{2})R_{QSP}D_j(\tfrac{\Delta^\prime}{2})$. The purpose of the displacement gate conjugation is to ensure that the Gaussian wavepacket at $x=0$ is contained entirely in the interval $(0,\infty)$. The inverse displacement moves all Gaussian wavepackets back to their original position.

The second stage is more involved. We begin by noting that we cannot simply apply the same set of QSP sequences as was used for the forward padding gate in Eq.~\ref{eq: forward padding GQSP Polynomial}, but with $a^\prime$ anti-padding qubits. The reason for this is that we would have to require that:
\begin{equation*}
\begin{split}
&R_{j,pad}^{(\Delta,a^\prime)}\ket{-2^{n/m}|\tilde{k^{(j)}}|+y^{(j)}}_{O_j}\ket{0}_{Q_j}\\
&= R_{j,pad}^{(\Delta,a^\prime)}\ket{2^{n/m}|\tilde{k^{(j)}}|+y^{(j)}}_{O_j}\ket{0}_{Q_j}
\end{split}
\end{equation*} 
for all $j, \tilde{k}^{(j)}$, which can be shown not to be the case. Instead, we note that the polynomial we need for $\tilde{k}^{(j)}\leq-1$ is a slightly modified version of $R_{j,pad}^{(\Delta,a^\prime)}$ that has been shifted by $(-2^{a^\prime-j}+1)2^{n/m}\Delta^\prime$. Thus, for negative $\tilde{k}^{(j)}$, what we need is:
\begin{equation}
\begin{split}
    \tilde{R}_{j,pad}^{(\Delta^\prime,a^\prime)}(\hat{x}) = R_{j,pad}^{(\Delta^\prime,a^\prime)}(\hat{x}+(2^{a^\prime-j}-1)2^{n/m}\Delta^\prime)
\end{split}
\end{equation}Information about the sign of $\tilde{k}^{(j)}$ is stored in the $a^\prime+1$th qubit register, therefore, we can apply two controlled QSP sequences; one in which $R_{j,pad}^{(\Delta,a^\prime)}$ is applied controlled on the $a^\prime+1$th qubit being $\ket{0}$ and one in which $\tilde{R}_{j,pad}^{(\Delta,a^\prime)}$ is applied controlled on the $a^\prime+1$th qubit being $\ket{1}$.
This concludes the construction of the entangling stage of the anti-padding gate. A circuit depiction of this gate is given in Fig.~\ref{fig:anti-padding gate entangling stage}.

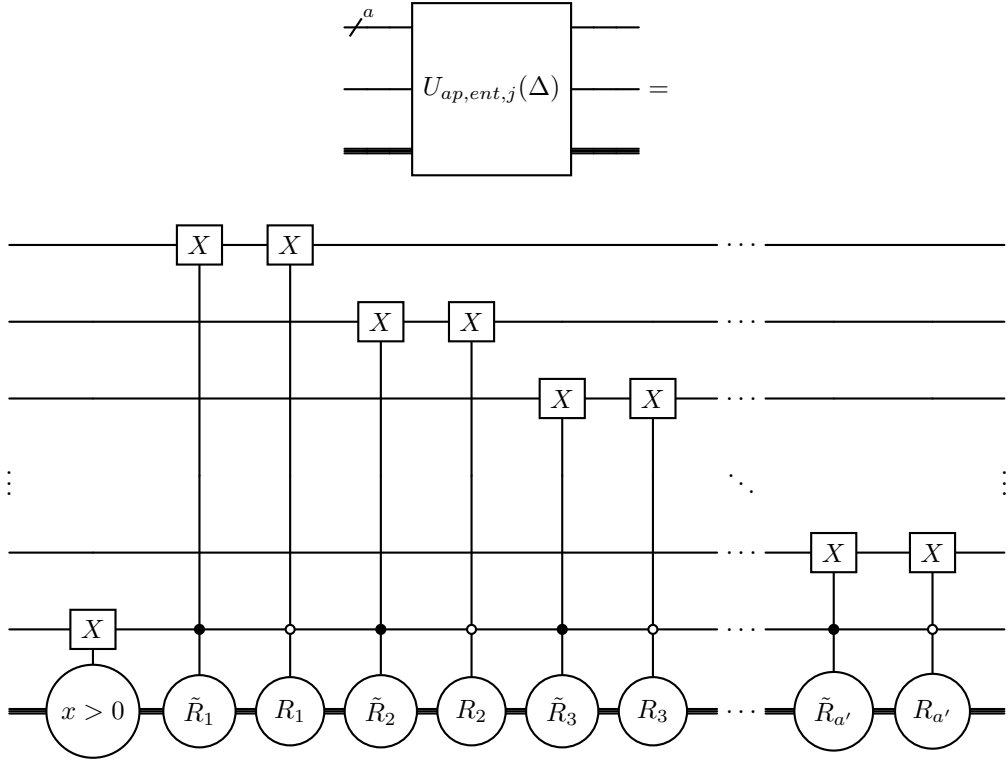
\begin{figure*}[ht!]
    \centering
    \begin{quantikz}[wire types={q,q,b}, column sep = 0.3cm]
            & \qwbundle[]{a} & & \gate[3]{U_{ap,ent,j}(\Delta)}& & &\\
            & & & & & &\\
            & & & & & &\\
        \end{quantikz} = 
        \begin{quantikz}[wire types={q,q,q,n,q,q,b}]
            & & \gate[]{X} & \gate[]{X} & & & & & \ \ldots \ & & &\\
            & & \wire[u]{q} & \wire[u]{q} & \gate[]{X} & \gate[]{X} & & & \ \ldots \ & & &\\
            & & \wire[u]{q} & \wire[u]{q} & \wire[u]{q} & \wire[u]{q} & \gate[]{X} & \gate[]{X} & \ \ldots \ & & &\\
            \vdots & & \wire[u]{q} & \wire[u]{q} & \wire[u]{q} & \wire[u]{q} & \wire[u]{q} & \wire[u]{q} &  \ \ddots \ & & & \vdots\\
            & & \wire[u]{q} & \wire[u]{q} & \wire[u]{q} & \wire[u]{q} & \wire[u]{q} & \wire[u]{q} & \ \ldots \ & \gate[]{X} & \gate[]{X} &\\
            & \gate[]{X} & \ctrl{1}\wire[u]{q} & \octrl{1}\wire[u]{q} & \ctrl{1}\wire[u]{q} & \octrl{1}\wire[u]{q} & \ctrl{1}\wire[u]{q} & \octrl{1}\wire[u]{q} &  \ \ldots \ & \ctrl{1}\wire[u]{q} &\octrl{1}\wire[u]{q} &\\
            & \gate[style=circle]{x>0}\wire[u]{q} & \gate[style=circle]{\tilde{R}_1} & \gate[style=circle]{R_1} & \gate[style=circle]{\tilde{R}_2} & \gate[style=circle]{R_2} & \gate[style=circle]{\tilde{R}_3} & \gate[style=circle]{R_3} &  \ \ldots \ & \gate[style=circle]{\tilde{R}_{a^\prime}} & \gate[style=circle]{R_{a^\prime}} & \\
        \end{quantikz}
    \caption{The entangling stage of the anti-padding gate.}
    \label{fig:anti-padding gate entangling stage}
\end{figure*}

\subsubsection{Constructing the second black box: $U_{ap,d,j}$}

How do we implement the disentangling gates $U_{ap,d,j}$? For this gate sequence, we require double-controlled displacement gates. That is, for $\ell_j=1,\dots,a^\prime$, we want the following gate:
\begin{equation}\label{eq: jth antipadding disentangler}
\begin{split}
    &U_{ap,d,j} = D_j(-2^{n/m+\ell_j-1}\Delta^\prime) \otimes \ket{q_1}_{Q}\ket{0}_{Q_{a^\prime}}\bra{q_1}_{Q}\bra{0}_{Q_{a^\prime}}\\
    &+ D_j(2^{n/m+\ell_j-1}\Delta^\prime)\otimes\ket{q_1}_{Q}\ket{1}_{Q_{a^\prime}}\bra{q_1}_{Q}\bra{1}_{Q_{a^\prime}}
\end{split}
\end{equation}

In practice, we can realize $U_{ap,d,j}$ in two stages: one corresponding to the first term in Eq.~\ref{eq: jth antipadding disentangler} and one corresponding to the second term. Each of these steps can be done with the addition of two Toffoli gates one additional ancilla qubit. In the first stage, we apply a Toffoli gate with the two data qubits as the control and the ancilla as the target. We then apply a conditional displacement of $-2^{n/m+j-1}\Delta$ conditioned on the ancilla qubit. We finish this stage by uncomputing the ancilla with another Toffoli gate. The second stage mirrors the first, with the only differences being that we must conjugate this circuit with a Pauli $X$ gate on the second data qubit and that the displacement we apply is $2^{n/m+\ell_j-1}\Delta$. The circuit for this operation is depicted in Fig.~\ref{fig:double-controlled displacement gate}, where we must run this gate for $p=0,1$. A circuit depiction of the disentangling stage as a whole is given in Fig.~\ref{fig:anti-padding gate disentangling stage}

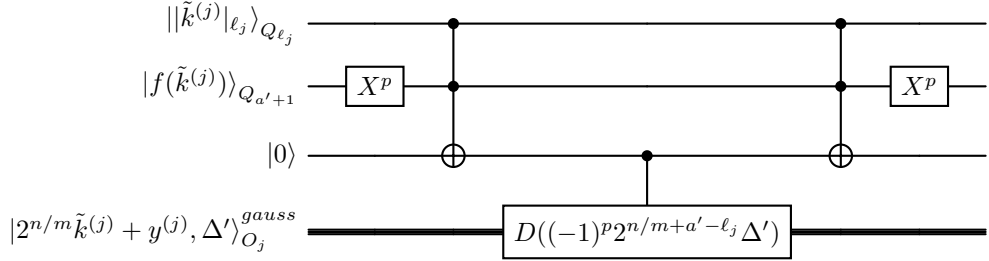
\begin{figure*}[ht!]
    \centering
    \begin{quantikz}[wire types={q,q,q,b}]
            \lstick{$\ket{|\tilde{k}^{(j)}|_{\ell_j}}_{Q_{\ell_j}}$}& & \ctrl{1} & & \ctrl{1} & &\\
            \lstick{$\ket{f(\tilde{k}^{(j)})}_{Q_{a^\prime+1}}$} & \gate[]{X^p} & \ctrl{1} & & \ctrl{1} & \gate[]{X^p} & \\
            \lstick{$\ket{0}$} & & \targ{} & \ctrl{1} & \targ{} & &\\
            \lstick{$\ket{2^{n/m}\tilde{k}^{(j)}+y^{(j)}, \Delta^\prime}^{gauss}_{O_j}$} & & & \gate[]{D((-1)^p2^{n/m+a^\prime-\ell_j}\Delta^\prime)}& & &\\
        \end{quantikz}
    \caption{The double-controlled displacement gate. The ancilla qubit in the third register stores the sum of the bits in the first two registers modulo 2. For all $\ell_j=1,2,\dots,a^\prime$, we run this circuit twice for $p=0,1$.}
    \label{fig:double-controlled displacement gate}
\end{figure*}

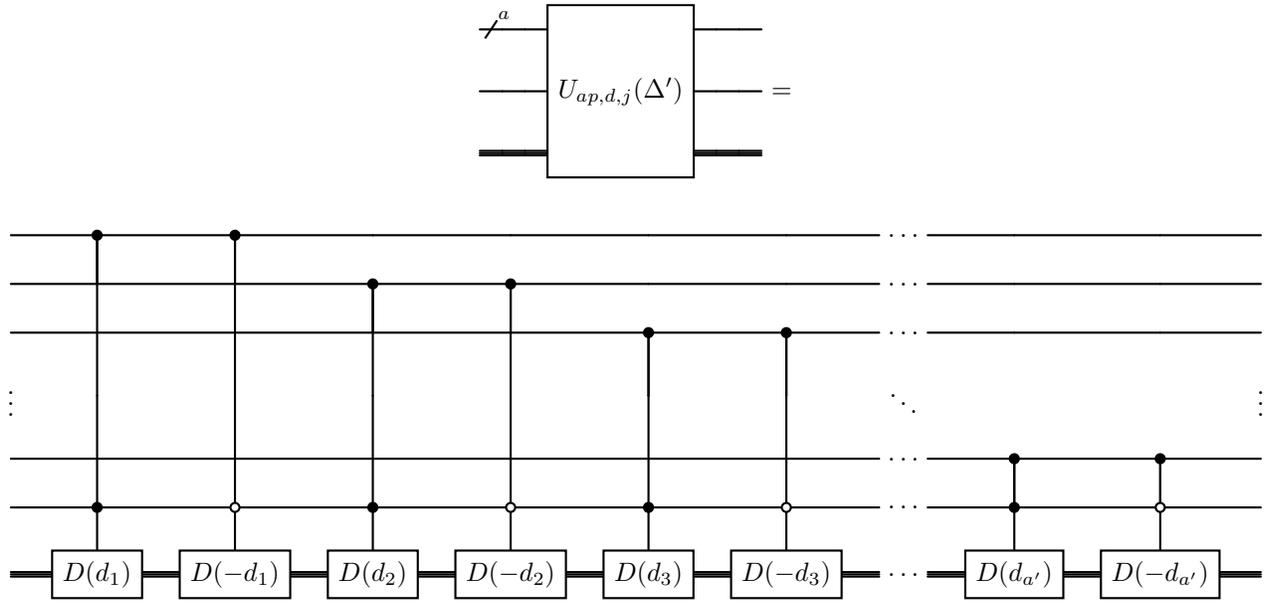
\begin{figure*}[ht!]
    \centering
    \begin{quantikz}[wire types={q,q,b}, column sep = 0.3cm]
            & \qwbundle[]{a} & & \gate[3]{U_{ap,d,j}(\Delta^\prime)}& & &\\
            & & & & & &\\
            & & & & & &\\
        \end{quantikz} = 
        \begin{quantikz}[wire types={q,q,q,n,q,q,b}]
            & \ctrl{1} & \ctrl{1} & & & & & \ \ldots \ & & &\\
            & \wire[u]{q} & \wire[u]{q} & \ctrl{1} & \ctrl{1} & & & \ \ldots \ & & &\\
            & \wire[u]{q} & \wire[u]{q} & \wire[u]{q} & \wire[u]{q} & \ctrl{1} & \ctrl{1} & \ \ldots \ & & &\\
            \vdots & \wire[u]{q} & \wire[u]{q} & \wire[u]{q} & \wire[u]{q} & \wire[u]{q} & \wire[u]{q} &  \ \ddots \ & & & \vdots\\
            & \wire[u]{q} & \wire[u]{q} & \wire[u]{q} & \wire[u]{q} & \wire[u]{q} & \wire[u]{q} & \ \ldots \ & \ctrl{1} & \ctrl{1} &\\
            & \ctrl{1}\wire[u]{q} & \octrl{1}\wire[u]{q} & \ctrl{1}\wire[u]{q} & \octrl{1}\wire[u]{q} & \ctrl{1}\wire[u]{q} & \octrl{1}\wire[u]{q} &  \ \ldots \ & \ctrl{1}\wire[u]{q} &\octrl{1}\wire[u]{q} &\\
            & \gate[]{D(d_1)} & \gate[]{D(-d_1)} & \gate[]{D(d_2)} & \gate[]{D(-d_2)} & \gate[]{D(d_3)} & \gate[]{D(-d_3)} &  \ \ldots \ & \gate[]{D(d_{a^\prime})} & \gate[]{D(-d_{a^\prime})} & \\
        \end{quantikz}
    \caption{The disentangling stage of the anti-padding gate. The displacement amounts $d_{\ell_j}$ are defined to be $2^{n/m+a^\prime-\ell_{j}}\Delta^\prime$.}
    \label{fig:anti-padding gate disentangling stage}
\end{figure*}Putting all of these steps together, we present a circuit for the Hadamard gate as a whole in Fig.~\ref{fig: Hadamard gate circuit}

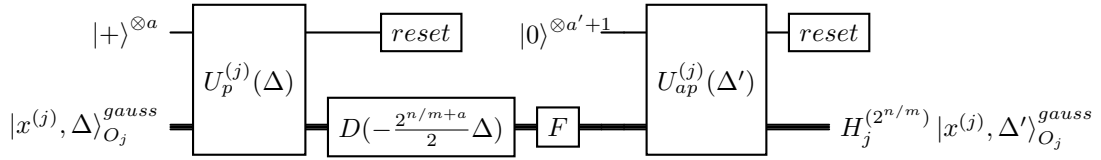
\begin{figure*}[ht!]
    \centering
    \begin{quantikz}[wire types={q,b}, column sep = 0.3cm]
        \lstick{$\ket{+}^{\otimes a}$}& \gate[2]{U_p^{(j)}(\Delta)} & \gate[]{reset} &\setwiretype{n}& &\lstick{$\ket{0}^{\otimes a^\prime+1}$} \setwiretype{q}&\gate[2]{U_{ap}^{(j)}(\Delta^\prime)} & \gate[]{reset}\\
        \lstick{$\ket{x^{(j)}, \Delta}_{O_j}^{gauss}$}& & \gate[]{D(-\tfrac{2^{n/m+a}}{2}\Delta)} & \gate[]{F} & & & & \rstick{$H^{(2^{n/m})}_j\ket{x^{(j)}, \Delta^\prime}_{O_j}^{gauss}$}\\
        \end{quantikz}
    \caption{The sequence of gate operations used to realize the qudit Hadamard on a qumode.}
    \label{fig: Hadamard gate circuit}
\end{figure*}

\subsection{Fidelity bounds on Hadamard gate}\label{sec: Hadamard gate error bounds}

A single application of the Hadamard gate involves the following steps which each (except the displacement gate step) introduce approximation errors that can be quantified:

\begin{enumerate}
    \item The application of the padding gate using $a$ padding qubits which incurs an error of $\mathcal{O}(a\epsilon)$ for $a$ padding qubits and a QSP error of $\epsilon$. 
    \item The error-free application of a displacement gate $D(2^{n/m+a-1}\Delta)$.
    \item The application of the free evolution gate, which incurs an approximation error of $\mathcal{O}(\tfrac{1}{2^a})$.
    \item The Gaussian envelope approximation in Eq.~\ref{eq: post Fourier gate H re-expressed} which incurs error $\mathcal{O}(\tfrac{\sigma}{\Delta})$.
    \item The approximation that Eq~\ref{eq: post Fourier gate H re-expressed} can be approximated as a finite sum over $\tilde{k}^{(j)}$. Below we show that this incurs an error of $\mathcal{O}(\mathrm{erfc}(\frac{\Delta}{\sigma}2^{a^\prime}))$ for $a^\prime$ anti-padding qubits. In order for this error to be on the order of machine epsilon, a good rule of thumb is to choose $a^\prime = \lceil\mathrm{log}_2(\frac{\Delta}{\sigma})+1\rceil$
    \item The anti-padding disentangling step which uses an approximate QSP circuit.
    
\end{enumerate}

The errors for the first four of these is given in ~\cite{liu2024mixed}, thus we will start by deriving the error bound of the fifth item. Eq~\ref{eq: post Fourier gate H re-expressed} describes the action of this sequence of gates on a single basis state $\ket{x^{(j)},\Delta}^{gauss}_{O_j}$. Consider now the action of this gate on a superposition of such inputs states with coefficients $C_{x^{(j)}}$. In the instance where this Hadamard gate is being applied to one qumode among several, these coefficients would be those obtained by tracing out all the qumodes on which this Hadamard gate acts trivially. The norm of these coefficients will be thus bounded above by 1. This state is given by:

\begin{widetext}
\begin{equation}
\begin{split}
    \sqrt{\frac{2\pi\sigma}{\Delta 2^{n/m}}\sqrt{\frac{2}{\pi}}}\sum_{\tilde{k}(j)=-2^{a^\prime}/2}^{2^{a^\prime}/2-1}\sum_{y^{(j)}=0}^{2^{n/m}-1}\sum_{x^{(j)}=0}^{2^{n/m}-1}C_{x^{(j)}}e^{2\pi ix^{(j)}y^{(j)}/2^{n/m}}e^{-(\tfrac{2\pi}{\Delta}\tilde{k}^{(j)})^2\sigma^2}\ket{2^{n/m}\tilde{k}^{(j)}+y^{(j)},\Delta^\prime}_{O_j}^{gauss}.
\end{split}
\end{equation}
\end{widetext}We can put an error bound on this approximation by considering the norm of the state generated by all the summation terms that were omitted in this truncation:
\begin{widetext}
\begin{equation}
\begin{split}
    &\left|2\sqrt{\frac{2\pi \sigma}{\Delta 2^{n/m}}\sqrt{\frac{2}{\pi}}}\sum_{\tilde{k}^{(j)}=2^{a^\prime}/2}^\infty\left[\sum_{x^{(j)}=0}^{2^{n/m}-1}C_{x^{(j)}}^*e^{2\pi ix^{(j)}y^{(j)}/2^{n/m}}\right]e^{-(\tfrac{2\pi}{\Delta}k^{(j)})^2\sigma^2}\ket{2^{n/m}\tilde{k}^{(j)}+y^{(j)},\Delta^\prime}^{gauss}_{O_j}\right|^2\\
    &=\frac{8\pi \sigma}{\Delta}\sqrt{\frac{2}{\pi}}\sum_{\tilde{k}^{(j)},\tilde{k}^{(j)\prime} = 2^{a^\prime}/2}^{\infty}e^{-(\tfrac{2\pi}{\Delta}\tilde{k}^{(j)})^2\sigma^2}e^{-(\tfrac{2\pi}{\Delta}\tilde{k}^{(j)\prime})^2\sigma^2}\\
    &\times\left[\frac{1}{2^{n/m}}\sum_{y^{(j)},y^{(j)\prime}=0}^{2^{n/m}-1}\sum_{x^{(j)},x^{(j)\prime}=0}^{2^{n/m}-1}C_{x^{(j)}}C_{x^{(j)\prime}}^*e^{2\pi i x^{(j)}y^{(j)}/2^{n/m}}e^{-2\pi ix^{(j)\prime}y^{(j)\prime}}\braket{2^{n/m}\tilde{k}^{(j)}+ y^{(j)\prime},\Delta^\prime|2^{n/m}\tilde{k}^{(j)}+y^{(j)},\Delta^\prime}_{O_j}^{gauss}\right].
\end{split}
\end{equation}
\end{widetext}We assume that:
\begin{equation*}
    \braket{2^{n/m}\tilde{k}^{(j)}+ y^{(j)\prime},\Delta^\prime|2^{n/m}\tilde{k}^{(j)}+y^{(j)},\Delta^\prime}_{O_j}^{gauss}
\end{equation*}is approximately a delta function, which is true when the variance of these peaks is much smaller than $\Delta^\prime$. This allows us to write:
\begin{widetext}
\begin{equation}
\begin{split}
    &\frac{8\pi \sigma}{\Delta}\sqrt{\frac{2}{\pi}}\sum_{\tilde{k}^{(j)},\tilde{k}^{(j)\prime}=2^{a^\prime}/2}^{\infty}e^{-(\tfrac{2\pi}{\Delta}\tilde{k}^{(j)})^2\sigma^2}e^{-(\tfrac{2\pi}{\Delta}\tilde{k}^{(j)\prime})^2\sigma^2}\\
    &\times\left[\frac{1}{2^{n/m}}\sum_{y^{(j)},y^{(j)\prime}=0}^{2^{n/m}-1}\sum_{x^{(j)},x^{(j)\prime}=0}^{\infty}C_{x^{(j)}}C_{x^{(j)\prime}}^*e^{2\pi i x^{(j)}y^{(j)}/2^{n/m}}e^{-2\pi i x^{(j)\prime}y^{(j)\prime}/2^{n/m}}\delta_{y^{(j)}y^{(j)\prime}}\delta_{\tilde{k}^{(j)}\tilde{k}^{(j)\prime}}\right]\\
    &=\frac{8\pi \sigma}{\Delta}\sqrt{\frac{2}{\pi}}\sum_{\tilde{k}^{(j)},\tilde{k}^{(j)\prime}=2^{a^\prime}/2}^{\infty}e^{-(\tfrac{2\pi}{\Delta}\tilde{k}^{(j)})^2\sigma^2}e^{-(\tfrac{2\pi}{\Delta}\tilde{k}^{(j)\prime})^2\sigma^2}\delta_{\tilde{k}^{(j)}\tilde{k}^{(j)\prime}}\\
    &\times\left[\frac{1}{2^{n/m}}\sum_{x^{(j)},x^{(j)\prime}=0}^{2^{n/m}-1}\sum_{y^{(j)}=0}^{2^{n/m}-1}C_{x^{(j)}}C_{x^{(j)\prime}}\left[\sum_{y^{(j)}=0}^{2^{n/m}-1}e^{2\pi iy^{(j)}(x^{(j)}-x^{(j)\prime})/2^{n/m}}\right]\right]\\
    &=\frac{8\pi \sigma}{\Delta}\sqrt{\frac{2}{\pi}}\sum_{\tilde{k}^{(j)},\tilde{k}^{(j)\prime}=2^{a^\prime}/2}^{\infty}e^{-(\tfrac{2\pi}{\Delta}\tilde{k}^{(j)})^2\sigma^2}e^{-(\tfrac{2\pi}{\Delta}\tilde{k}^{(j)\prime})^2\sigma^2}\delta_{\tilde{k}^{(j)}\tilde{k}^{(j)\prime}}\left[\frac{1}{2^{n/m}}\sum_{x^{(j)},x^{(j)\prime}=0}^{2^{n/m}-1}C_{x^{(j)}}C_{x^{(j)\prime}}\left[2^{n/m}\delta_{x^{(j)}x^{(j)\prime}}\right]\right]\\
    &=\frac{8\pi \sigma}{\Delta}\sqrt{\frac{2}{\pi}}\sum_{\tilde{k}^{(j)},\tilde{k}^{(j)\prime}=2^{a^\prime}/2}^{\infty}e^{-(\tfrac{2\pi}{\Delta}\tilde{k}^{(j)})^2\sigma^2}e^{-(\tfrac{2\pi}{\Delta}\tilde{k}^{(j)\prime})^2\sigma^2}\delta_{\tilde{k}^{(j)}\tilde{k}^{(j)\prime}}\left[\sum_{x^{(j)}=0}^{2^{n/m}-1}|C_{x^{(j)}}|^2\right]\\
    &\leq \frac{8\pi \sigma}{\Delta}\sqrt{\frac{2}{\pi}}\sum_{\tilde{k}^{(j)},\tilde{k}^{(j)\prime}=2^{a^\prime}/2}^{\infty}e^{-(\tfrac{2\pi}{\Delta}\tilde{k}^{(j)})^2\sigma^2}e^{-(\tfrac{2\pi}{\Delta}\tilde{k}^{(j)\prime})^2\sigma^2}\delta_{\tilde{k}^{(j)}\tilde{k}^{(j)\prime}}\\
    &=\frac{8\pi \sigma}{\Delta}\sqrt{\frac{2}{\pi}}\sum_{\tilde{k}^{(j)}=2^{a^\prime}}^{\infty}e^{-(\tfrac{2\sqrt{2}\pi}{\Delta}\tilde{k}^{(j)})^2\sigma^2}.
\end{split}
\end{equation}
\end{widetext}The sum over $\tilde{k}^{(j)}$ is the sum of the areas of rectangles with width 1 and height $e^{-(\tfrac{2\sqrt{2}\pi}{\Delta}\tilde{k}^{(j)})^2\sigma^2}$. Thus, we can bound this summation by an indefinite integral over the interval $[2^{a^\prime}-1,\infty)$, giving us an upper bound on the error of:
\begin{equation}
\begin{split}
    &\leq \frac{8\pi\sigma}{\Delta}\sqrt{\frac{2}{\pi}}\int_{2^{a^\prime}-1}^\infty e^{-(\tfrac{2\sqrt{2}\pi}{\Delta}x)^2\sigma^2}dx\\
    &=2\left(\mathrm{erfc}\left(\tfrac{2\sqrt{2}\pi\sigma}{\Delta}\left(\tfrac{2^{a^\prime}}{2}-1\right)\right)\right).
\end{split}
\end{equation}In the limit of $2^{a^\prime-1}\gg1$, this is approximately:
\begin{equation}
\begin{split}
    \approx2\left(\mathrm{erfc}(\sqrt{2}\pi\tfrac{\sigma}{\Delta}2^{a^\prime})\right)
\end{split}
\end{equation}To get a sense of a scale for how large $a^\prime$ has to be in order for the error contribution for this particular approximation step to be negligible, we can observe that $2(\mathrm{erfc(1}))\approx0.3$ and $2(\mathrm{erfc}(6))\approx10^{-17}$. That is, once $2^{a^\prime}$ surpasses $\frac{1}{\sqrt{2}\pi}\frac{\Delta}{\sigma}$, the error drops exponentially and quickly becomes negligible. Thus, we require that $a^\prime = \mathcal{O}(\mathrm{log_2}(\frac{\Delta}{\sigma}))$ up to an additive constant factor that is on the order of unity.

\section{Correcting position displacement errors on a cat qubit}\label{sec: correct cat state}

The error suppression scheme described above essentially suppresses position displacement errors on a $d=2^n$ qudit whose levels are encoded into Gaussian wavepackets located on a 1D lattice of the position of the oscillator. This requires $\mathrm{log}_2(d)$ controlled displacement gates and QSP sequences. The largest displacement required is $\tfrac{d\Delta}{2}$ and the circuit depth of the highest QSP circuit needed is $\mathcal{O}(\tfrac{d}{\Delta}\mathrm{log}_2(\tfrac{1}{\epsilon}))$. Therefore, we expect this to be most readily applied to problems where $d$ is small, $\Delta$ is not small, and $\epsilon$ does not necessarily need to be small. One problem that fits all of these criteria is suppressing displacement errors on a cat qubit:
\begin{equation}
\begin{split}
    \ket{\psi_{cat}} &= C_0\ket{0}_L + C_1\ket{1}_L\\
    &=\tfrac{C_0}{\mathcal{N}_0}\left(\ket{\alpha}+\ket{-\alpha}\right) + \tfrac{C_1}{\mathcal{N}_1}\left(\ket{\alpha}-\ket{-\alpha}\right)\\
    &=\left(\tfrac{C_0}{\mathcal{N}_0}+\tfrac{C_1}{\mathcal{N}_1}\right)\ket{\alpha}+\left(\tfrac{C_0}{\mathcal{N}_0}-\tfrac{C_1}{\mathcal{N}_1}\right)\ket{-\alpha}.
\end{split}
\end{equation}This is essentially the $d=2$, $\Delta=2\alpha$ case. Suppose an unknown displacement $\delta$ occurs on the cat qubit to give us the state:
\begin{equation}
\begin{split}
    D(\delta)\ket{\psi_{cat}} &=\left(\tfrac{C_0}{\mathcal{N}_0}+\tfrac{C_1}{\mathcal{N}_1}\right)\ket{\alpha+\delta}+\left(\tfrac{C_0}{\mathcal{N}_0}-\tfrac{C_1}{\mathcal{N}_1}\right)\ket{-\alpha +\delta}.
\end{split}
\end{equation}Extending generalized QSP to hybrid systems in the future may allow us to implement polynomials of neither even or odd parity, however for now we restrict ourselves to standard single-variable hybrid QSP where we can only implement polynomials of definite parity. This requires us to apply a displacement of $D(\alpha)$ to give us the state:
\begin{equation}
\begin{split}
    &D(\alpha)D(\delta)\ket{\psi_{cat}}\\
    &=\left(\tfrac{C_0}{\mathcal{N}_0}+\tfrac{C_1}{\mathcal{N}_1}\right)\ket{2\alpha+\delta}+\left(\tfrac{C_0}{\mathcal{N}_0}-\tfrac{C_1}{\mathcal{N}_1}\right)\ket{\delta}.
\end{split}
\end{equation}Assume that $\delta < 2\alpha$ and that the variance of the coherent states $\ket{\pm\alpha}$ are much smaller than $2\alpha$. Introduce an ancilla qubit initialized as $\ket{0}$ and apply a hybrid QSP sequence which implements a step function centered at position $x=\alpha$. This gives us the state:
\begin{equation}
\begin{split}
     &R_{QSP}D(\alpha)D(\delta)\ket{\psi_{cat}}_O\ket{0}_Q\\
     &= \left(\tfrac{C_0}{\mathcal{N}_0}+\tfrac{C_1}{\mathcal{N}_1}\right)\ket{2\alpha+\delta}_O\ket{1}_Q+\left(\tfrac{C_0}{\mathcal{N}_0}-\tfrac{C_1}{\mathcal{N}_1}\right)\ket{\delta}_O\ket{0}_Q.
\end{split}
\end{equation}We now apply a conditional displacement of $-2\alpha$ to the qumode to give us:
\begin{equation}
\begin{split}
    &cD(-2\alpha)R_{QSP}D(\alpha)D(\delta)\ket{\psi_{cat}}_O\ket{0}_Q\\
    &= \left(\tfrac{C_0}{\mathcal{N}_0}+\tfrac{C_1}{\mathcal{N}_1}\right)\ket{\delta}_O\ket{1}_Q+\left(\tfrac{C_0}{\mathcal{N}_0}-\tfrac{C_1}{\mathcal{N}_1}\right)\ket{\delta}_O\ket{0}_Q.
\end{split}
\end{equation}Next we reset the qumode, giving us:
\begin{equation}
\begin{split}
    &G_{RESET}cD(-2\alpha)R_{QSP}D(\alpha)D(\delta)\ket{\psi_{cat}}_O\ket{0}_Q\\
    &= \ket{0}_O\left[\left(\tfrac{C_0}{\mathcal{N}_0}+\tfrac{C_1}{\mathcal{N}_1}\right)\ket{1}_Q+\left(\tfrac{C_0}{\mathcal{N}_0}-\tfrac{C_1}{\mathcal{N}_1}\right)\ket{0}_Q\right].
\end{split}
\end{equation}Next we apply a conditional displacement of $2\alpha$, yielding:
\begin{equation}
\begin{split}
    &cD(2\alpha)G_{RESET}cD(-2\alpha)R_{QSP}D(\alpha)D(\delta)\ket{\psi_{cat}}_O\ket{0}_Q\\
    &=\left(\tfrac{C_0}{\mathcal{N}_0}+\tfrac{C_1}{\mathcal{N}_1}\right)\ket{2\alpha}_O\ket{1}_Q+\left(\tfrac{C_0}{\mathcal{N}_0}-\tfrac{C_1}{\mathcal{N}_1}\right)\ket{0}_O\ket{0}_Q.
\end{split}
\end{equation}Applying the inverse QSP sequence $R_{QSP}^\dagger$ followed by a displacement of $-\alpha$ gives us:

\begin{widetext}
\begin{equation}
\begin{split}
    R_{QSP}^\dagger cD(2\alpha)G_{RESET}cD(-2\alpha)R_{QSP}D(\alpha)D(\delta)\ket{\psi_{cat}}_O\ket{0}_Q &= \ket{0}_Q\left[\left(\tfrac{C_0}{\mathcal{N}_0}+\tfrac{C_1}{\mathcal{N}_1}\right)\ket{\alpha}_O+\left(\tfrac{C_0}{\mathcal{N}_0}-\tfrac{C_1}{\mathcal{N}_1}\right)\ket{-\alpha}_O\right]\\
    &=\ket{0}_Q\ket{\psi_{cat}}_O.
\end{split}
\end{equation}
\end{widetext}This shows that provided that we can implement the QSP sequence and controlled displacement with high fidelity, we can correct an unknown displacement $\delta$ on the cat qubit provided that $\delta$ is smaller than $\alpha$ and that the variance of the coherent states $\ket{\pm \alpha}$ are small compared to $\alpha$.

\bibliographystyle{unsrt}
\bibliography{ref}
\end{document}